\documentclass[journal]{IEEEtran}

\usepackage{amsmath,amssymb,amsfonts}
\usepackage{mathtools}
\usepackage{amsthm}
\usepackage{xcolor}
\usepackage{cite}
\usepackage{graphicx}
\usepackage{booktabs}
\usepackage{url}
\usepackage{algorithm}
\usepackage{float}

\newtheorem{lemma}{Lemma}
\newtheorem{theorem}{Theorem}
\newtheorem{corollary}{Corollary}
\newtheorem{proposition}{Proposition}
\newtheorem{remark}{Remark}

\newcommand{\E}{\mathbb E}

\newcommand{\tr}{\operatorname{tr}}
\newcommand{\vecop}{\operatorname{vec}}

\newcommand{\diag}{\operatorname{diag}}

\newcommand{\tfracorall}{\forall}

\newcommand{\nw}[1]{#1}
\newenvironment{nwblock}{}{}

\begin{document}

\title{{ Data-Driven LQR with Uncertainty Refinement and
Domain Randomization---Part~II: Fisher-Majorized Certified Policy Iteration}}

\author{Hamidreza Modares %
\thanks{H. Modares is with Michigan State University,
East Lansing, MI 48824 USA.}}

\maketitle

\begin{abstract}
\begin{nwblock}
Building on the uncertainty description developed in Part~I, this
paper introduces a Fisher-majorized domain-randomized generalized policy
iteration framework for data-driven linear quadratic regulator (LQR) synthesis. At each iteration, partial
Bellman evaluation over a fixed batch of uncertainty realizations produces a
structured Sylvester policy-improvement step. A residual-weighted Fisher curvature controls the policy-induced variation in the stationary covariance and converts the Sylvester model into a certified upper bound on the cost change, thereby guaranteeing descent for every accepted step. A check-and-correct mechanism then certifies stability over the
entire uncertainty support, rather than only over the sampled realizations. Consequently, every iterate is robustly stabilizing and monotonically decreases the fixed-batch objective. The Fisher truncation and stability correction 
become inactive after finitely many times. The objective gap then contracts linearly to a floor determined by the dispersion of the realization-wise Riccati targets and by
measurement noise. In the absence of measurement noise, this floor vanishes
when the realization-wise targets coincide.
\end{nwblock}
\end{abstract}


\section{Introduction}
\label{sec:introduction}
\begin{nwblock}
\IEEEPARstart{T}{his} paper uses the uncertainty description of Part~I~\cite{PartI} to synthesize controllers for the linear quadratic regulator (LQR) under uncertainties. Robust LQR
typically optimizes the worst admissible model and often uses a common
performance--stability certificate
\cite{Robust1,Ma2025RobustLQR,PolyakTempo2001,DUST,8933093,Robust2,LAI2023110685,Floriancov1,DEPERSIS2021109548,UncertaintyLearning4},
which can be conservative or infeasible for large uncertainty sets. Domain
randomization (DR) instead optimizes expected performance
\cite{Fujinami2025DRLQR,Fujinami2025PolicyGradient,nguyenle2026globally,Pasdar2026PI} to reduce conservatism on perforamnce,
but sample-wise stability does not certify the unsampled support.
\end{nwblock}

This paper introduces a Fisher-majorized domain-randomized generalized policy iteration framework for data-driven LQR synthesis. At each iteration, partial Bellman evaluation over a fixed batch of uncertainty realizations yields a Sylvester policy-improvement step. A residual-weighted Fisher curvature is introduced to policy-induced variations in the stationary covariance. A check-and-correct mechanism then enforces stability over the entire uncertainty support, rather than only over the sampled realizations: steps that leave the certified controller polytope are corrected within it, while Fisher-truncated steps may be extended along the same stability-certified ray through a safeguarded one-dimensional search. Consequently, every iterate robustly stabilizes the full support and monotonically decreases the fixed-batch objective. The line search is invoked only finitely many times. Under exact evaluation and the stated interiority condition, both Fisher truncation and stability correction eventually become inactive. The objective gap then contracts linearly to a floor determined by the dispersion of the realization-wise Riccati improvement directions and measurement noise. In the absence of measurement noise, this floor vanishes when the realization-wise directions coincide.

Compared with robust LQR, FM-DR-GPI reduces conservatism by optimizing expected performance using cell-dependent stability certificates rather than a common worst-case certificate. Unlike existing DR-LQR methods, it guarantees full-support stability independently of the sampled realizations and certifies descent through Fisher-based majorization. Moreover, constrained stationarity is established without gradient dominance or small-heterogeneity assumptions, even with nonzero measurement noise.

\vspace{3pt}
\noindent \textbf{Notations:}
\nw{We follow Part~I~\cite{PartI}. For a set $\mathcal S$,
$\operatorname{co}(\mathcal S)$, $\operatorname{vert}(\mathcal S)$, and
$\operatorname{int}(\mathcal S)$ are its convex hull, vertices, and interior;
$\mathsf N_{\mathcal C}(k)$ is the normal cone. The inverse column-stacking
maps are $\vecop$ and $\operatorname{mat}$, $[m]:=\{1,\ldots,m\}$,
$\rho(\cdot)$ is spectral radius, and $\subsetneq$ is strict inclusion.}
\vspace{-6pt}

\section{Problem Formulation}
\subsection{Part~I Model and Certified Uncertainty}
\label{sec:recap}
Consider the linear system of Paper I \cite{PartI} as
\begin{equation}
x_{k+1}=A^\star x_k+B^\star u_k+\omega_k,\qquad y_k=x_k+\upsilon_k,
\label{eq:syst}
\end{equation}


Let $\bar\Phi=\operatorname{col}(\bar Y_0,\bar U_0)$ and $\bar Y_0$  and $\bar U_0$ be the collected normalized output and input data, respectively. Let also  $S_x:=\operatorname{col}(I_n,0)$,
$S_u:=\operatorname{col}(0,I_m)$,
$T_x:=\bar\Phi^{-1}S_x$, and $T_u:=\bar\Phi^{-1}S_u$. Then, for $u_k=K x_k$, 
\nw{Part~I parametrizes each control gain $K$ one-to-one by} a decision variable $G_K$ as
\begin{equation}
G_K:=T_x+T_uK
=\bar\Phi^{-1}\operatorname{col}(I_n,K),
\label{Gk-K}
\end{equation}
\nw{It satisfies $\bar Y_0G_K=I_n$ and gives the closed-loop matrix}
\begin{equation}
A_K(\eta):=A(\eta)+B(\eta)K
=(\bar Y_1-N(\eta))G_K,
\label{eq:DR-closed-loop-GK}
\end{equation}
\nw{where $\bar Y_1$ shifts $\bar Y_0$ by one step.

Part~I also returns a certified constrained matrix zonotope that covers its refined model set,
\begin{align}
\Theta_r
&:=
\{\theta(\eta):\eta\in\mathcal P\},\nonumber\\
\mathcal P
&:=
\left\{
\eta\in\mathbb R^{q_\theta}:
\|\eta\|_\infty\le1,\;
A_\eta\eta=b_\eta
\right\},
\label{eq:Pj-explicit}
\end{align}

Here, $\Theta_r$ denotes the constrained-matrix-zonotope cover, which contains the exact refined model set of Part~I, so $\theta^\star\in\Theta_r$ holds on the confidence event of Part~I. Each
$\eta\in\mathcal P$ jointly determines 
\begin{align}
\theta(\eta)
&:=
C_\theta+\sum_{s=1}^{q_\theta}G_{\theta,s}\eta_s
=:[A(\eta)\ B(\eta)],\nonumber\\
N(\eta)
&:=
\bar Y_1-\theta(\eta)\bar\Phi
=
C_N+\sum_{s=1}^{q_\theta}G_{N,s}\eta_s,
\label{thetaN}
\end{align}
where
\(
C_N:=\bar Y_1-C_\theta\bar\Phi,
\)
and
\(
G_{N,s}:=-G_{\theta,s}\bar\Phi.
\)

For $\omega_k\sim\mathcal N(0,W)$ and $\upsilon_k\sim\mathcal N(0,V)$, Part~I also provides certified bounds
}
\begin{equation}
0\preceq W\preceq\bar W,\quad \bar W\succ0,\qquad
0\preceq V\preceq\bar V.
\label{eq:WV-certificates}
\end{equation}
Here, $\bar W$ and $\bar V$ are any valid certificates from Part~I. All results below hold for any such choice.

 \vspace{-10pt}

\subsection{{Uncertainty-Partitioned  Domain-Randomized LQR}}
\label{sec:formulation}
\label{sec:LQRform}
\nw{For $Q\succeq0$, $R\succ0$, $u_k=Ky_k$, and
$\theta=[\,A\;B\,]$, define}
\begin{align} \label{eq:lqrcost}
J(K,\theta)
=
\lim_{\tau\to\infty}
\frac{1}{\tau}
\sum_{k=0}^{\tau}
\mathbb E\!\left[
y_k^\top Qy_k+u_k^\top Ru_k
\right].
\end{align}
\nw{under the standard stabilizability of $(A,B)$ and detectability of
$(A,Q^{1/2})$. $K$ acts on the noisy full-state measurement.} \vspace{3pt}


\noindent\textbf{Problem 1:}
Given the uncertainty description \eqref{eq:Pj-explicit}, partition its
coefficient domain into $\{\mathcal P_\ell\}_{\ell=1}^{L}$ and find a single feedback gain $K$ that solves
\begin{align}
\label{DR-problem}
& \min_{K}\quad
J_{\rm DR}(K)
:=
\mathbb E_{\eta\sim \mathbb P_{\rm des}}
\!\left[J(K,\theta(\eta))\right]
\nonumber\\
\text{s.t.}\quad
&
A_K(\eta)\mathcal X_\ell
\subseteq
\lambda_\ell\mathcal X_\ell,
\quad
\forall\,\eta\in\mathcal P_\ell,
\quad \ell=1,\ldots,L,
\end{align}
\nw{where $\mathbb P_{\rm des}$ weights performance and the regional
contractions certify the complete uncertainty support.}
\hfill $\square$

\nw{We optimize a sample-average approximation (SAA) and then transfer its
guarantees to Problem~1.} \vspace{-3pt}

\section{Robust-Stability Certification}
\label{subsec:DR-partition}
Polyhedral sets provide a natural representation for partitioning the learned uncertainty support. Therefore, the coefficient domain in \eqref{eq:Pj-explicit} is covered by closed polytopic
cells \vspace{-6pt}
\begin{equation*}
\mathcal P=\bigcup_{\ell=1}^{L}\mathcal P_\ell,
\qquad
\mathcal P_\ell:=\{\eta:F_\ell\eta\le g_\ell\},
\end{equation*}
with pairwise disjoint relative interiors.  Let $\operatorname{vert}(\mathcal P_\ell)=\{\eta_{\ell v}\}_{v=1}^{N_\ell^\theta}$
and define $N_{\ell v}:=N(\eta_{\ell v})$ as the reconstructed
noise matrix at vertex $\eta_{\ell v}$ of the uncertainty
cell $\mathcal P_\ell$.  Since $N(\eta)$ and
$\theta(\eta)$ in \eqref{thetaN} are affine, 
\begin{equation*}
\mathcal N_\ell=N(\mathcal P_\ell)
=\operatorname{co}\{N_{\ell v}\}_{v=1}^{N_\ell^\theta},
\qquad
\Theta_\ell:=\theta(\mathcal P_\ell).
\end{equation*}

Besides, polyhedral Lyapunov sets can reduce the geometric conservatism of quadratic certificates~\cite{sojoudani}. As shown next, however, a common certificate can be highly conservative. 
 \vspace{3pt} 

\noindent\textbf{Motivating example:}
Consider the uncertain system
\(
x_{k+1}=A(\eta)x_k+u_k,
\,
\eta\in[0,1],
\)
where
\[
A(\eta)=
\begin{bmatrix}
0.4 & 1-\eta\\
\eta & 0.4
\end{bmatrix}.
\]
Let $K=0$. Then, the closed-loop eigenvalues are
\(
0.4\pm\sqrt{\eta(1-\eta)}.
\)
Since
\(
0\leq\sqrt{\eta(1-\eta)}\leq 0.5,
\)
every fixed realization is Schur stable.
A single contraction certificate, however, cannot certify this gain over the
whole interval. Indeed, at the endpoints,
\(
\rho\!\left(A(0)A(1)\right)
\approx1.30>1.
\)
If there existed a common contraction set $\mathcal X$ and some
$\lambda<1$ such that
\(
A(\eta)\mathcal X
\subseteq
\lambda\mathcal X,
\,
\forall\eta\in[0,1],
\)
then, in particular,
\(
A(0)A(1)\mathcal X
\subseteq
\lambda^2\mathcal X,
\)
 implying
$\rho(A(0)A(1))\leq\lambda^2<1$, which is a contradiction.
Now, consider
\(
[0,1]
=
[0,0.35]\cup[0.35,0.5]
\cup[0.5,0.65]\cup[0.65,1].
\)
For the four cells, respectively, define the polyhedral certificates
\(
\mathcal X_j
=
\{x:\|D_jx\|_{\infty}\leq1\},
\)
where
\(
D_1=\operatorname{diag}(0.595,1),
D_2=\operatorname{diag}(0.875,1),
D_3=\operatorname{diag}(1.14,1),
D_4=\operatorname{diag}(1,0.595).
\)
Since $D_j=\operatorname{diag}(a_j,b_j)$, direct calculation gives
\(
\|D_jA(\eta)D_j^{-1}\|_\infty
=
\max\left\{
0.4+\frac{a_j}{b_j}(1-\eta),
\ 0.4+\frac{b_j}{a_j}\eta
\right\}
\leq 0.995<1
\)
on each corresponding cell, where the maximum is attained at a
cell endpoint. Thus, $V_j(x)=\|D_jx\|_\infty$ satisfies
\(
V_j(A(\eta)x)\leq 0.995V_j(x),
\)
or equivalently,
\(
A(\eta)\mathcal X_j\subseteq 0.995\,\mathcal X_j
\)
for every $\eta$ in the cell. 
 
 \vspace{6pt}

We precompute finitely many candidate certificate families, e.g., by
counterexample-guided synthesis \cite{sojoudani}. Family \(j\) assigns to each
uncertainty cell \(\mathcal P_\ell\) a polyhedral Lyapunov set
\(\mathcal X_\ell^{(j)}\) and a contraction factor
\(\lambda_\ell^{(j)}\in(0,1)\). That is,
\begin{equation} \label{hset}
\begin{aligned}
\mathcal H^{(j)}
:=
\{\mathcal X_\ell^{(j)},\lambda_\ell^{(j)}\}_{\ell=1}^{L}, \quad 0<\lambda_\ell^{(j)}<1 \\
\mathcal X_\ell^{(j)}
:=
\{x:\mathsf H_\ell^{(j)}x\le\mathbf1\},
\, \, \mathsf H_\ell^{(j)\top}
=
\big[
h_{\ell 1}^{(j)} & \cdots & h_{\ell m_{\ell,j}}^{(j)}
\big].
\end{aligned}
\end{equation}

 Given polyhedral certificates and polyhedral partitioned uncertainties, the following lemma gives a characterization of the stability constraints in \eqref{DR-problem}.

\begin{lemma}
\label{thm:DR-robust-stability}
For the set $\mathcal X_\ell^{(j)}$ in \eqref{hset} with $0\in\operatorname{int}\mathcal X_\ell^{(j)}$, write
\(
\operatorname{vert}(\mathcal X_\ell^{(j)})
=
\{\xi_{\ell q}^{(j)}\}_{q=1}^{N_{\ell,j}^{x}}.
\)  Given \eqref{Gk-K}-\eqref{eq:DR-closed-loop-GK}, define 
\begin{equation*}
\mathcal U_{\ell v}
:=\left[\,M_{\ell v}T_x\quad M_{\ell v}T_u\,\right], \quad M_{\ell v}:=\bar Y_1-N_{\ell v},
\end{equation*}
where $N_{\ell v}:=N(\eta_{\ell v})$. Define its extreme-image index set as
\begin{equation}
\widehat{\mathcal V}_\ell
:=
\left\{
v\in[N_\ell^\theta]:
\mathcal U_{\ell v}\in
\operatorname{vert}\!\left(
\operatorname{co}\{
\mathcal U_{\ell s}:s\in[N_\ell^\theta]
\}
\right)
\right\}.
\label{eq:DR-control-image-hull}
\end{equation}
{For each retained tuple $i=(\ell,v,r,q)$, define}
\begin{align*}
g_i^{(j)}&:=\xi_{\ell q}^{(j)}\otimes
\left(T_u^\top M_{\ell v}^\top h_{\ell r}^{(j)}\right),\\
b_i^{(j)}&:=\lambda_\ell^{(j)}-
 h_{\ell r}^{(j)\top}M_{\ell v}T_x\xi_{\ell q}^{(j)},
\end{align*}
and let
\(
\mathcal I_j^{\rm raw}
:=\big\{(\ell,v,r,q):\;\ell\in[L],\
 v\in\widehat{\mathcal V}_\ell,
 r\in[m_{\ell,j}],\ q\in[N_{\ell,j}^{x}]\big\}.
\)
Then, the set of gains satisfying the constraints of \eqref{DR-problem} for the $\ell$-th cell and for the certificate $\mathcal X_\ell^{(j)}$ are exactly
\begin{equation}
\mathcal C_j^{\rm raw}
:=\left\{k:g_i^{(j)\top}k\le b_i^{(j)},\,
 i\in\mathcal I_j^{\rm raw}\right\},
\, k=\vecop(K).
\label{eq:DR-raw-controller-polyhedron}
\end{equation}
Every $k\in\mathcal C_j^{\rm raw}$ robustly stabilizes
the closed-loop system.
\end{lemma}

\begin{proof}
  See Appendix~\ref{app:DR-robust-stability}. 
\end{proof}

Fix a compact polyhedral controller domain
\(
\mathcal K_0:=\{k:F_0k\le g_0\}\subset\mathbb R^{d_K},
\, d_K:=mn,
\)
containing every controller permitted by the algorithm.  It may be fixed from
known gain or actuator limits.

\begin{proposition}
\label{prop:DR-normal-fan-reduction}
For fixed
\(
\ell\in[L]
\),
\(
v\in\widehat{\mathcal V}_\ell
\),
and
\(
r\in[m_{\ell,j}]
\),
define
\(
c_{\ell vr}^{(j)}(K)
:=G_K^\top M_{\ell v}^\top h_{\ell r}^{(j)},\,
\mathcal D_{\ell vr}^{(j)}
:=\{c_{\ell vr}^{(j)}(K):\vecop(K)\in\mathcal K_0\}.
\)
Let $\mathsf N_{\mathcal X_\ell^{(j)}}(\xi_{\ell q}^{(j)})$ denote the normal cone of $\mathcal X_\ell^{(j)}$ at $\xi_{\ell q}^{(j)}$, and define
\begin{align} \label{newQ}
\widehat{\mathcal Q}_{\ell vr}^{(j)}
:=\left\{q\in[N_{\ell,j}^{x}]:
\mathsf N_{\mathcal X_\ell^{(j)}}(\xi_{\ell q}^{(j)})
\cap\mathcal D_{\ell vr}^{(j)}\neq\varnothing\right\}.
\end{align}
For every \(\vecop(K)\in\mathcal K_0\),
\begin{equation}
\max_{q\in[N_{\ell,j}^{x}]}
c_{\ell vr}^{(j)}(K)^\top\xi_{\ell q}^{(j)}
=
\max_{q\in\widehat{\mathcal Q}_{\ell vr}^{(j)}}
c_{\ell vr}^{(j)}(K)^\top\xi_{\ell q}^{(j)}.
\label{eq:DR-normal-fan-support-equality}
\end{equation}
Consequently, all controller rows indexed by
\(q\notin\widehat{\mathcal Q}_{\ell vr}^{(j)}\) may be omitted without
changing the certified controller set on \(\mathcal K_0\).
\end{proposition}

\begin{proof}
  See Appendix~\ref{app:DR-normal-fan-reduction}. 
\end{proof} \vspace{-3pt}

By Lemma~\ref{thm:DR-robust-stability} and
Proposition~\ref{prop:DR-normal-fan-reduction}, the only retained
uncertainty-image indices are $v\in\widehat{\mathcal V}_\ell$, and the only
relevant certificate vertices are
$q\in\widehat{\mathcal Q}_{\ell vr}^{(j)}$.  Thus, for bundle $j$, define the
pre-reduction constraint index set
$\mathcal I_j^{\rm pre}:=\{(\ell,v,r,q):\ell\in[L],\,v\in\widehat{\mathcal V}_\ell,\,r\in[m_{\ell,j}],\,q\in\widehat{\mathcal Q}_{\ell vr}^{(j)}\}$. We then remove redundant controller inequalities sequentially.  At a given
stage, let $\mathcal I\subseteq\mathcal I_j^{\rm pre}$ denote the current
retained index set and let $i\in\mathcal I$.  The implication LP obtained by
omitting row $i$ is
\begin{equation*}
\rho_i^{(j)}(\mathcal I)
:=
\sup_{k\in\mathcal K_0}
\left\{
g_i^{(j)\top}k:
g_s^{(j)\top}k\le b_s^{(j)},
\ s\in\mathcal I\setminus\{i\}
\right\}.
\end{equation*}
which gives the standard exact
redundancy test: sequentially remove row $i$ whenever
$\rho_i^{(j)}(\mathcal I)\le b_i^{(j)}$, and denote the remaining certificate
rows by $\mathcal I_j^{\rm red}$.  Each removal is logically implied by the
currently retained rows on $\mathcal K_0$, so the feasible set is unchanged.
We therefore define
\begin{equation}
\mathcal C_j:=\{k\in\mathcal K_0:
 g_i^{(j)\top}k\le b_i^{(j)},\ i\in\mathcal I_j^{\rm red}\}.
\label{Cfinal1}
\end{equation}
For subsequent parent-set computations, write equivalently
\begin{align} \label{Cfinal}
  \mathcal C_j=\{k:G_jk\le h_j\}  
\end{align}
and let $\mathcal I_j^+$ index all rows of this stacked representation.

\vspace{-9pt}


\subsection{Bounded-Complexity Inner Controller Sets}
\label{subsec:DR-inner-controller-sets}
{If the facet count of $\mathcal C_j$ in \eqref{Cfinal} is acceptable, we use
$\mathcal C_j$ directly. Otherwise, the next lemma provides the exact radial
reach for generating certified boundary points and a strictly interior anchor
for the initial inner polytope.} 

\begin{lemma}
\label{prop:DR-parent-geometry}
Let $\operatorname{int}\mathcal C_j\neq\varnothing$ and for $k\in\mathcal C_j$ in \eqref{Cfinal} and $d\neq0$, define
\begin{equation}
\bar\alpha_j(k,d)
:=
\max_{\alpha\ge0}
\left\{
\alpha:
G_j(k+\alpha d)\le h_j
\right\}.
\label{eq:DR-parent-radial-reach}
\end{equation}
Let $G_{j,i}$ and $h_{j,i}$ denote the $i$th row
of $G_j$ and the corresponding entry of $h_j$,
respectively, for $i\in\mathcal I_j^+$. Then, 
\begin{equation}
\bar\alpha_j(k,d)
=
\min_{\substack{
i\in\mathcal I_j^{+}:
G_{j,i}d>0}}
({h_{j,i}-G_{j,i}k})/({G_{j,i} \, d}).
\label{eq:DR-parent-radial-reach-explicit}
\end{equation}
Moreover, let
\begin{align}
&(a_j,r_j)\in\arg\max_{a,r\ge0}\quad r
\nonumber\\
&\qquad
\text{\rm s.t.}\quad
G_{j,i}a+r\|G_{j,i}\|_2\le h_{j,i},
\,
i\in\mathcal I_j^{+}.
\label{eq:DR-inner-anchor-LP}
\end{align}
{Then, $r_j>0$, $a_j\in\operatorname{int}\mathcal C_j$, and
$\mathbb B_2(a_j,r_j)\subseteq\mathcal C_j$ is the largest Euclidean ball
contained in $\mathcal C_j$.}
\end{lemma}

\begin{proof}
 See Appendix~\ref{app:DR-parent-geometry}.   
\end{proof}


We now construct a
bounded-complexity inner approximation of
\eqref{Cfinal}. Let $d_K:=mn$ and choose $p_{\max}\ge d_K+1$. Let $a_j$ be the center of the largest Euclidean ball
inscribed in $\mathcal C_j$, obtained from
Lemma~\ref{prop:DR-parent-geometry}.
We first construct an initial full-dimensional inner polytope by
selecting $d_K+1$ unit directions
$\{u_{j,r}\}_{r=1}^{d_K+1}$ that form a  simplex
in $\mathbb R^{d_K}$, whose vertices are equally spaced
around the origin, as
\begin{equation}
\sum_{r=1}^{d_K+1}u_{j,r}=0,
\quad
u_{j,r}^{\top}u_{j,s}=-\frac{1}{d_K},
\quad r\neq s.
\label{eq:DR-simplex-directions}
\end{equation}
Starting from the interior point $a_j$, extend each direction
to the boundary of $\mathcal C_j$ using
Lemma~\ref{prop:DR-parent-geometry} as
\(
\tau_{j,r}:=\bar\alpha_j(a_j,u_{j,r}),
\,
\kappa_{j,r}^{(0)}
:=a_j+\tau_{j,r}u_{j,r}.
\)
The convex hull of these boundary points defines the initial
full-dimensional certified inner polytope,
\begin{equation}
\widehat{\mathcal C}_{j,0}
:=
\operatorname{co}\{\kappa_{j,r}^{(0)}\}_{r=1}^{d_K+1}.
\label{eq:DR-inner-controller-polytope}
\end{equation}

Now, at iteration $t$, let
\(
\widehat V_{j,t}:=
\big[\kappa_{j,1}^{(t)}\ \cdots\
\kappa_{j,p_{j,t}}^{(t)}\big],
\quad p_{j,t}\le p_{\max},
\)
collect the vertices of the current inner polytope.
For $k_t\in\widehat{\mathcal C}_{j,t}$ and $d\neq0$,
the largest step along $d$ that remains in the inner polytope
is obtained by
\begin{equation}
\begin{aligned}
&\widehat\alpha_{j,t}(d)
:=\max_{0\le\alpha\le1,\gamma}\quad\alpha\\
\text{\rm s.t.}\quad&
k_t+\alpha d=\widehat V_{j,t}\gamma,\quad
\mathbf1^\top\gamma=1,\quad
\gamma\ge0.
\end{aligned}
\label{eq:DR-inner-bundle-directional-margin}
\end{equation}
Here, $\gamma$ represents the convex-combination weights
of the stored vertices, while $\alpha\le1$ limits the
search to the proposed step.
Compare this with the exact feasible step in the parent
polytope $\mathcal C_j$,
\(
\alpha_{j,t}^{\rm ex}(d):=
\min\{1,\bar\alpha_j(k_t,d)\}.
\)
If the inner polytope is overly restrictive along $d$,
as measured by
\begin{equation}
\alpha_{j,t}^{\rm ex}(d)-\widehat\alpha_{j,t}(d)
>\varepsilon_{\rm cmp},
\qquad
p_{j,t}<p_{\max},
\label{eq:DR-inner-enrichment-trigger}
\end{equation}
we enlarge it by adding the certified endpoint
\begin{equation}
\kappa_{j,t}^{\rm new}
:=
k_t+\alpha_{j,t}^{\rm ex}(d)d,
\qquad
\widehat{\mathcal C}_{j,t}^{+}
:=
\operatorname{co}\!\left(
\widehat{\mathcal C}_{j,t}\cup
\{\kappa_{j,t}^{\rm new}\}
\right).
\label{eq:DR-inner-enrichment}
\end{equation}
The next result confirms that enrichment expands the inner polytope in the
desired direction while preserving
$\widehat{\mathcal C}_{j,t}^{+}\subseteq\mathcal C_j$.
\begin{proposition}
\label{thm:DR-bounded-complexity-inner-set}
The construction
\eqref{eq:DR-simplex-directions}--\eqref{eq:DR-inner-controller-polytope}
satisfies
\(
a_j\in\operatorname{int}\widehat{\mathcal C}_{j,0},
\,
\widehat{\mathcal C}_{j,0}\subseteq\mathcal C_j,
\)
and, for every subsequent inner polytope,
\(
\widehat\alpha_{j,t}(d)\le\alpha_{j,t}^{\rm ex}(d).
\)
After enrichment steps,  one has
\begin{equation}
\widehat{\mathcal C}_{j,t}
\subsetneq
\widehat{\mathcal C}_{j,t}^{+}
\subseteq
\mathcal C_j,
\qquad
\widehat\alpha_{j,t}^{+}(d)=\alpha_{j,t}^{\rm ex}(d).
\label{eq:DR-inner-enrichment-properties}
\end{equation}
\end{proposition}
\begin{proof}
See Appendix~\ref{app:DR-bounded-complexity-inner-set}.
\end{proof}

 The exact set \eqref{Cfinal}, the inner ball $\mathbb B_2(a_j,r_j)$, and the initial inner polytopes are computed offline.



\vspace{-9pt}

\section{Certified DR Objective and Fisher Geometry}
\label{sec:DR-LQR}
This section constructs the objective and realization-wise Fisher quantities
used by FM-DR-GPI.

\vspace{-9pt}
\subsection{Realization-Dependent LQR Cost}
\label{sec:DR-cost}
For \(\eta\in\mathcal P\) and \(K\) such that \(A_K(\eta)\) in
\eqref{eq:DR-closed-loop-GK} is Schur, let \(P_K(\eta)\succeq0\) be the
unique solution of the Lyapunov equation
\begin{equation}
P_K(\eta)
=
A_K(\eta)^\top P_K(\eta)A_K(\eta)+S_K,
\label{eq:DR-pointwise-P}
\end{equation}
where  \vspace{-6pt}
\begin{equation}
S_K:=Q+K^\top RK,
\label{Sk}
\end{equation}
Define  \vspace{-6pt}
\begin{equation}
\Omega_K(\eta)
:=
\bar W+B(\eta)K\bar V K^\top B(\eta)^\top.
\label{eq:DR-Omega}
\end{equation}
We next provide a realization-dependent surrogate LQR cost.

\begin{lemma}
\label{lem:certified-pointwise-cost}
For every
Schur pair $(K,\eta)$, let 
\begin{equation}
\bar J(K,\theta(\eta))
:=
\operatorname{tr}\!\left(P_K(\eta)\Omega_K(\eta)\right)
+\operatorname{tr}(S_K\bar V).
\label{eq:DR-pointwise-cost}
\end{equation}
Then, $J(K,\theta(\eta))$ defined in \eqref{eq:lqrcost} satisfies
\begin{equation}
J(K,\theta(\eta))
\leq
\bar J(K,\theta(\eta)).
\label{eq:pointwise-cost-bound}
\end{equation}
Moreover, if the certified bounds are consistent, i.e., $\bar W_T\to W$ and $\bar V_T\to V$, then
\begin{equation}
\lim_{T\to\infty}
\sup_{\eta\in\mathcal P_T}
\left|
\bar J^{(T)}(K,\theta(\eta))-J(K,\theta^\star)
\right|
=0.
\label{eq:pointwise-cost-asymptotic}
\end{equation}
\end{lemma}
\begin{proof}
See Appendix~\ref{app:certified-pointwise-cost}.    
\end{proof}

Motivated by Lemma~\ref{lem:certified-pointwise-cost}, we replace the
unknown cost \(J(K,\theta(\eta))\) in \eqref{DR-problem} by \(\bar J(K,\theta(\eta))\). Thus, the
DR objective of Problem 1 used for synthesis is
\(
\bar J_{\rm DR}(K)
:=
\mathbb E_{\eta\sim\mathbb P_{\rm des}}
\!\left[\bar J(K,\theta(\eta))\right].
\)
Drawing a batch
 of independent realizations  with fixed weight $1/N$ (i.e., uniform distribution), the sample-average approximation (SAA) objective is
\begin{equation}
f_N(K):=\frac{1}{N} \sum_{i=1}^{N}\bar J(K,\theta(\eta_i)).
\label{eq:DR-GPI-finite-objective}
\end{equation}

\vspace{-9pt}

\subsection{Fisher Geometry of Closed-Loop Policy Perturbations}
Let $K_t:=K^{(t)}$ denote the controller gain at iteration $t$,
and define $A_t(\eta):=A_{K_t}(\eta)$ and
$\Omega_t(\eta):=\Omega_{K_t}(\eta)$, with $K$ being replaced with $K_t$ in \eqref{eq:DR-closed-loop-GK} and \eqref{eq:DR-Omega}, respectively.
Let $X_t(\eta)$ denote the surrogate stationary
covariance, distinguished from the true covariance
$X_K^{\rm true}$ in Lemma~\ref{lem:certified-pointwise-cost},
and satisfying
\begin{equation}
X_t(\eta)=\Omega_t(\eta)
+A_t(\eta)X_t(\eta)A_t(\eta)^\top.
\label{eq:DR-Fisher-stationary-X}
\end{equation}
Since the closed-loop dynamics are
$x_{k+1}=A_t(\eta)x_k+B(\eta)K_t\upsilon_k+\omega_k$,
the corresponding Gaussian surrogate transition, using
the covariance bounds of Lemma~\ref{lem:certified-pointwise-cost},
is
\begin{equation}
p_{K_t}^\eta(\cdot\mid x)
:=
\mathcal N\!\left(A_t(\eta)x,\Omega_t(\eta)\right).
\label{eq:DR-Fisher-transition}
\end{equation}
For a policy perturbation $K(\varepsilon):=K_t+\varepsilon D$,
the directional derivatives of the transition mean and covariance are
\begin{align}
\dot\Omega_t(\eta;D)
&=B(\eta)D\bar V K_t^\top B(\eta)^\top
+B(\eta)K_t\bar V D^\top B(\eta)^\top,
\nonumber\\
\dot\mu_t(\eta,x;D)
&=
\big.
\frac{d}{d\varepsilon}
\mu_{K_t+\varepsilon D}(\eta)
\big|_{\varepsilon=0}=B(\eta)Dx.
\label{eq:DR-Fisher-cov-derivative}
\end{align}
For a Gaussian variable
$z\sim\mathcal N(\mu(\vartheta),\Sigma(\vartheta))$,
the Fisher quadratic form along a parameter perturbation is
\(
\dot\mu^\top\Sigma^{-1}\dot\mu
+\frac12\tr\!\big(
\Sigma^{-1}\dot\Sigma\,
\Sigma^{-1}\dot\Sigma
\big).
\)
Substituting \eqref{eq:DR-Fisher-cov-derivative} into this
expression and averaging over the stationary state, using
$\mathbb E[xx^\top]=X_t(\eta)$, gives
\begin{align}
&\mathfrak I_{K_t}^{\rm tr}(\eta;D)
:=
\tr\!\left(
\Omega_t(\eta)^{-1}B(\eta)D X_t(\eta)
D^\top B(\eta)^\top
\right)
\nonumber\\
&\quad+
\frac12\tr\!\left(
\Omega_t(\eta)^{-1}\dot\Omega_t(\eta;D)
\Omega_t(\eta)^{-1}\dot\Omega_t(\eta;D)
\right).
\label{eq:DR-Fisher-directional-form}
\end{align}

The Gaussian transition law then
satisfies 
\begin{align*} \mathbb E_x[D_{\rm KL}(p_{K_t}^\eta\|p_{K(\varepsilon)}^\eta)]
=\frac{\varepsilon^2}{2}\mathfrak I_{K_t}^{\rm tr}(\eta;D)+o(\varepsilon^2),
\end{align*}
where  $p_{K(\varepsilon)}^\eta(\cdot\mid x)$ is defined as in
\eqref{eq:DR-Fisher-transition}, with $K^{(t)}$ replaced by $K(\varepsilon)$.
Thus, $\mathfrak I_{K_t}^{\rm tr}(\eta;D)$ measures, to second order in the step size
$\varepsilon$, how much the transition law of realization $\eta$ changes when
the gain moves along $D$. 


The Fisher sensitivity \eqref{eq:DR-Fisher-directional-form} is not directly
computable over the support. Thus, the certificate uses, for $K\in\mathcal S_{\rm cert}$ and
$d=\vecop(D)$, the quadratic form
\begin{equation}
\begin{aligned}
\mathfrak I_K(\eta;D)&:=d^\top\mathbf F_K(\eta)d,\\
\mathbf F_K(\eta)&:=(X_K(\eta)+\bar V)\otimes B(\eta)^\top\bar W^{-1}B(\eta).
\end{aligned}
\label{eq:DR-certified-Fisher-form}
\end{equation}

\begin{lemma}
\label{lem:DR-dominating-Fisher}
For every $\eta\in\mathcal P$ and every direction $D$, one has
\begin{equation}
\mathfrak I_{K_t}^{\rm tr}(\eta;D)\le\mathfrak I_{K_t}(\eta;D).
\label{eq:DR-dominating-Fisher}
\end{equation}
where $\mathfrak I_{K_t}^{\rm tr}(\eta;D)$ is defined in \eqref{eq:DR-Fisher-directional-form},
\end{lemma}
\begin{proof}
See Appendix \ref{proof:lem4}.    
\end{proof}  \vspace{-9pt}

\section{{Fisher-Majorized Certified GPI}}
\label{sec:SAC-PI}

\nw{Section~\ref{sec:DR-LQR} supplied the SAA objective, support-wide
Fisher form, and certified controller sets. We now connect them in three
stages: derive the exact frozen-covariance Sylvester model, convert it into a
computable upper model using partial evaluation and Fisher covariance control,
and check and correct its step for stability and descent.} \vspace{-12pt}

\subsection{Exact one-step policy-improvement model}
\nw{For a sampled $\eta_{\ell i}^{(t)}$, set
$A_{\ell i}^{(t)}:=A(\eta_{\ell i}^{(t)})$,
$B_{\ell i}^{(t)}:=B(\eta_{\ell i}^{(t)})$,
$F_{\ell i}^{(t)}(K):=A_{\ell i}^{(t)}+B_{\ell i}^{(t)}K$,
$F_{\ell i,t}:=F_{\ell i}^{(t)}(K^{(t)})$, and
$S^{(t)}:=S_{K^{(t)}}$. Also define
$\Omega_{\ell i}^{(t)}(K):=\bar W+B_{\ell i}^{(t)}K\bar V K^\top
B_{\ell i}^{(t)\top}$. The current-policy value matrix satisfies}
\begin{equation}
P_{\ell i}^{(t)}
=
S^{(t)}
+
F_{\ell i,t}^{\top}
P_{\ell i}^{(t)}
F_{\ell i,t}.
\label{eq:DR-GPI-current-value-equation}
\end{equation}

\nw{Apply a candidate $K$ for one step and then return to $K^{(t)}$. Under
the surrogate covariances,
$\E[y_k^\top Qy_k+u_k^\top Ru_k\mid x_k]
=x_k^\top S_Kx_k+\tr(S_K\bar V)$; hence the corresponding current-policy
Q-function is}
\begin{align}
& Q_{\ell i}^{(t)}(x_k;K)
={}
x_k^\top S_Kx_k+\tr(S_K\bar V)
\nonumber\\
&+
x_k^\top F_{\ell i}^{(t)}(K)^\top P_{\ell i}^{(t)}
F_{\ell i}^{(t)}(K)x_k
+\tr\!\left(P_{\ell i}^{(t)}
\Omega_{\ell i}^{(t)}(K)\right).
\label{eq:DR-GPI-candidate-Q}
\end{align}
\nw{Because the continuation matrix is fixed at $P_{\ell i}^{(t)}$,
\eqref{eq:DR-GPI-candidate-Q} is quadratic in $K$.}
Thus, the statewise one-step advantage is
\begin{align}
&a_{\ell i}^{(t)}(x;K)
:={}
Q_{\ell i}^{(t)}(x;K)
-
Q_{\ell i}^{(t)}(x;K^{(t)})
\nonumber\\ &
={}
x^\top[S_K-S^{(t)}]x
+
\tr\!\left([S_K-S^{(t)}]\bar V\right)
\nonumber\\
&+
x^\top
\left[
F_{\ell i}^{(t)}(K)^\top
P_{\ell i}^{(t)}
F_{\ell i}^{(t)}(K)
-
F_{\ell i,t}^\top
P_{\ell i}^{(t)}
F_{\ell i,t}
\right]x
\nonumber\\
&+
\tr\!\left(
P_{\ell i}^{(t)}
[
\Omega_{\ell i}^{(t)}(K)
-
\Omega_{\ell i}^{(t)}(K^{(t)})
]
\right).
\label{eq:DR-GPI-statewise-advantage}
\end{align}
\nw{Average this advantage under the current-policy state distribution:}
\begin{equation}
\mathcal A_{\ell i}^{(t)}(K)
:=
\mathbb E[
a_{\ell i}^{(t)}(x;K)].
\label{eq:DR-GPI-expected-advantage-definition}
\end{equation}
\nw{Using $\E[x^\top Mx]=\tr(MX_{\ell i}^{(t)})$ and
\eqref{eq:DR-GPI-current-value-equation} gives}
\begin{align*}
\mathcal A_{\ell i}^{(t)}(K)
={}&
\tr\!\left(
\left[
S_K
+F_{\ell i}^{(t)}(K)^\top P_{\ell i}^{(t)}
 F_{\ell i}^{(t)}(K)
-P_{\ell i}^{(t)}
\right]X_{\ell i}^{(t)}
\right)
\nonumber\\
&+
\tr\!\left([S_K-S^{(t)}]\bar V\right)
\nonumber\\
&+
\tr\!\left(
P_{\ell i}^{(t)}
[
\Omega_{\ell i}^{(t)}(K)
-\Omega_{\ell i}^{(t)}(K^{(t)})
]
\right).
\end{align*}
\nw{With uniform weights $w_{\ell i}=1/N$, define}
\begin{equation}
\mathcal A_N(K)
:=
\sum_{\ell=1}^{L}\sum_{i=1}^{N_\ell}
w_{\ell i}
\mathcal A_{\ell i}^{(t)}(K).
\label{eq:DR-GPI-weighted-advantage}
\end{equation}
\nw{Its current-policy gradient equals the SAA gradient:}

\begin{equation}
\nabla_K\mathcal A_N(K^{(t)})
=
\nabla_K f_N(K^{(t)})
=:g_t^N,
\label{eq:DR-GPI-advantage-gradient}
\end{equation}
\nw{where $f_N$ is given by \eqref{eq:DR-GPI-finite-objective} and}
\begin{equation}
\resizebox{0.96\columnwidth}{!}{$
\begin{aligned}
g_t^N
=2\sum_{\ell=1}^{L}\sum_{i=1}^{N_\ell} w_{\ell i}\Big[
&\left(R+B_{\ell i}^{(t)\top}P_{\ell i}^{(t)}B_{\ell i}^{(t)}\right)
K^{(t)}\left(X_{\ell i}^{(t)}+\bar V\right)\\[-1mm]
&+B_{\ell i}^{(t)\top}P_{\ell i}^{(t)}A_{\ell i}^{(t)}X_{\ell i}^{(t)}\Big].
\end{aligned}$}
\label{eq:DR-GPI-explicit-gradient}
\end{equation}
\nw{Because the averaging law is that of $K^{(t)}$, $\mathcal A_N$ is a
first-order model of $f_N(K)-f_N(K^{(t)})$, not the exact finite difference.}

\begin{lemma}
\label{lem:DR-GPI-exact-Sylvester}
With $P_{\ell i}^{(t)}$ and $X_{\ell i}^{(t)}$ fixed at the current policy,
$\mathcal A_N$ in \eqref{eq:DR-GPI-weighted-advantage} is a convex quadratic in $K$ and
\begin{equation}
\mathcal A_N(K^{(t)}+D)
=
\langle g_t^N,D\rangle
+
\frac{1}{2}\langle D,\mathcal M_t[D]\rangle,
\label{eq:DR-GPI-quadratic-model}
\end{equation}
with
\begin{align}
& \mathcal M_t[D]
= \nonumber \\ &
2\sum_{\ell=1}^{L}\sum_{i=1}^{N_\ell}
w_{\ell i}
\left(
R+B_{\ell i}^{(t)\top}P_{\ell i}^{(t)}B_{\ell i}^{(t)}
\right)
D
\left(X_{\ell i}^{(t)}+\bar V\right).
\label{eq:DR-GPI-M-operator}
\end{align}
Thus, the unconstrained policy-improvement direction is obtained from \vspace{-6pt}
\begin{equation}
\mathcal M_t[D_t^{\rm unc}]=-g_t^N.
\label{eq:DR-GPI-unconstrained-Sylvester}
\end{equation}
Moreover, the solution is unique and hence the
unique unconstrained minimizer is
$K_{\rm unc}^{(t+1)}=K^{(t)}+D_t^{\rm unc}$.
\end{lemma}

\begin{proof}
Let $\mathbf M_t^N$ denote the matrix representation of $\mathcal M_t$. Since
$R+B_{\ell i}^\top P_{\ell i}^{(t)}B_{\ell i}\succ0$ and
$X_{\ell i}^{(t)}+\bar V\succ0$, we have $\mathbf M_t^N\succ0$.
Hence, $K_{\rm unc}^{(t+1)}=K^{(t)}+D_t^{\rm unc}$ is the unique
unconstrained minimizer. Equation~\eqref{eq:DR-GPI-quadratic-model} follows by expanding the weighted
quadratic advantage \eqref{eq:DR-GPI-weighted-advantage} about $K^{(t)}$. Its
first derivative is $g_t^N$ by \eqref{eq:DR-GPI-advantage-gradient}, and its
Hessian action is \eqref{eq:DR-GPI-M-operator}. Setting the derivative with
respect to $D$ to zero gives \eqref{eq:DR-GPI-unconstrained-Sylvester}. 
\end{proof} \vspace{-12pt}


\subsection{Partial Evaluation and Certified Upper Model}
\label{subsec:DR-Fisher-partial-improvement}
\nw{The exact Sylvester model requires exact Lyapunov solutions and freezes
the current covariance. Finite Bellman sweeps remove the first requirement.
Certificate and Fisher bounds then control the omitted covariance change,
turning an e-modified quadratic into a computable upper model of the SAA cost
change. That model supplies the direction and a certified step; a scalar search
may extend it along the same stable ray.}

\subsubsection{Partial evaluation}
\label{subsubsec:DR-partial-evaluation}
\nw{Fix the batch across iterations, write
$\Omega_{\ell i,t}:=\Omega_{\ell i}^{(t)}(K_t)$, initialize
$P_{\ell i,t}^{(0)}=S^{(t)}$ and
$X_{\ell i,t}^{(0)}=\Omega_{\ell i,t}$, and perform $m_t\ge1$ sweeps}
\begin{align}
P_{\ell i,t}^{(j+1)}
&=S^{(t)}+F_{\ell i,t}^{\top}P_{\ell i,t}^{(j)}F_{\ell i,t},
\label{eq:DR-GPI-P-sweep}\\
X_{\ell i,t}^{(j+1)}
&=\Omega_{\ell i,t}+F_{\ell i,t}X_{\ell i,t}^{(j)}F_{\ell i,t}^{\top}.
\label{eq:DR-GPI-X-sweep}
\end{align}
Set
\begin{equation}
\widetilde P_{\ell i,t}:=P_{\ell i,t}^{(m_t)},
\qquad
\widetilde X_{\ell i,t}:=X_{\ell i,t}^{(m_t)}.
\label{eq:DR-GPI-tracked-PX}
\end{equation}
\nw{Substitution into \eqref{eq:DR-GPI-explicit-gradient} gives}
\begin{align*}
\widetilde g_t
=2\sum_{\ell=1}^{L}\sum_{i=1}^{N_\ell}w_{\ell i}
\Big[&
\left(R+B_{\ell i}^{(t)\top}\widetilde P_{\ell i,t}B_{\ell i}^{(t)}\right)
K^{(t)}\left(\widetilde X_{\ell i,t}+\bar V\right)
\nonumber\\[-1mm]
&+B_{\ell i}^{(t)\top}\widetilde P_{\ell i,t}
A_{\ell i}^{(t)}\widetilde X_{\ell i,t}\Big].
\end{align*}
\nw{The partially evaluated vectorized model is}
\begin{equation} \label{partialq}
{\widetilde{\mathfrak q}_t(d):=\widetilde{\mathbf g}_t^\top d
+\frac{1}{2} d^\top\widetilde{\mathbf M}_t d},
\qquad d:=\vecop(D),
\end{equation}
\nw{where $\widetilde{\mathbf g}_t:=\vecop(\widetilde g_t)$ and
$\vecop(\widetilde{\mathcal M}_t[D])=\widetilde{\mathbf M}_t d$, with}
\begin{align*}
& \widetilde{\mathcal M}_t[D]
= \nonumber \\ & 2\sum_{\ell=1}^{L}\sum_{i=1}^{N_\ell}w_{\ell i}
\left(R+B_{\ell i}^{(t)\top}\widetilde P_{\ell i,t}B_{\ell i}^{(t)}\right)
 D\left(\widetilde X_{\ell i,t}+\bar V\right).
\end{align*}
\nw{The recursions preserve
$\widetilde P_{\ell i,t}\succeq0$ and
$\widetilde X_{\ell i,t}\succeq\bar W$; hence, for
$\sigma_0:=\lambda_{\min}(R)\lambda_{\min}(\bar W)>0$,}
\(
\widetilde{\mathbf M}_t\succeq2\sigma_0 I_{d_K}.
\)

\subsubsection{Certified covariance bounds and Fisher region}
\label{subsubsec:DR-Fisher-region}

\nw{The model above freezes $X_K(\eta)$, whereas the new cost uses
$X_{K+D}(\eta)$. Fisher sensitivity controls this sole covariance mismatch,
but first the regional contraction certificates must bound $X_K(\eta)$ and
$P_K(\eta)$ uniformly.}
Consider the compact set \vspace{-9pt}
\begin{equation} \label{Scert}
\mathcal S_{\rm cert}
:=
\Big\{K:\vecop(K)\in\bigcup_{j=1}^{J_c}\mathcal C_j\Big\}.
\end{equation}
\nw{For each retained bundle $j$ and cell $\ell$, use the facet--vertex
representation}
\(
\mathcal X_\ell^{(j)}
=\{x:h_{\ell r}^{(j)\top}x\le1,\ 
r\in [m_{\ell,j}]\}
=\operatorname{conv}\{\xi_{\ell q}^{(j)}:
q\in [N_{\ell,j}^{x}] \}.
\)
\nw{Choose stored-geometry radii $r_{\ell j},R_{\ell j}>0$ such that
$r_{\ell j}\mathbb B_2\subseteq\mathcal X_\ell^{(j)}
\subseteq R_{\ell j}\mathbb B_2$, and set}
\(
\kappa_\ell
:=\max_j\frac{R_{\ell j}}{r_{\ell j}},
\lambda_\ell
:=\max_j\lambda_\ell^{(j)}<1,
\bar k_2
:=\max_{k\in\mathcal K_0}\|\operatorname{mat}(k)\|_2,
\bar b_\ell
:=\max_v\|B(\eta_{\ell v})\|_2,
\bar\omega_\ell
:=\|\bar W\|_2+\bar b_\ell^2\bar k_2^2\|\bar V\|_2,
\bar s
:=\|Q\|_2+\|R\|_2\bar k_2^2.
\)
\nw{Here, $\bar\omega_\ell$ and $\bar s$ bound
$\|\Omega_K(\eta)\|_2$ and $\|S_K\|_2$ for the evaluation-error analysis. Also
set}
\begin{equation}
\begin{aligned}
\bar\omega_{\ell j}
&:=\max_{v,r}\Big(h_{\ell r}^{(j)\top}\bar Wh_{\ell r}^{(j)}
+\bar k_2^2\|\bar V\|_2\big\|B(\eta_{\ell v})^\top h_{\ell r}^{(j)}\big\|_2^2\Big),\\
\bar s_{\ell j}
&:=\max_q\Big(\xi_{\ell q}^{(j)\top}Q\,\xi_{\ell q}^{(j)}
+\|R\|_2\bar k_2^2\|\xi_{\ell q}^{(j)}\|_2^2\Big),
\end{aligned}
\label{eq:DR-polyhedral-forcing}
\end{equation}
\nw{where $\eta_{\ell v}$ are the retained vertices of $\mathcal P_\ell$.}

\begin{lemma}
\label{prop:DR-certificate-computable-bounds}
Consider the closed-loop dynamics \eqref{eq:DR-closed-loop-GK}, let
$K_t\in\mathcal S_{\rm cert}$, $\eta\in\mathcal P_\ell$, and
$\mathcal J_t:=\{j:\vecop(K_t)\in\mathcal C_j\}$. Then, 
\begin{equation}
\|A_{K_t}(\eta)^k\|_2\le\kappa_\ell\lambda_\ell^k.
\label{eq:DR-certificate-power-bound}
\end{equation}
for every $k\ge0$. Moreover, for $j\in\mathcal J_t$, the current contraction factor and forcing
level
\begin{equation}
\begin{aligned}
\lambda_{\ell j}^{\rm cl}(K_t)
&:=\max_{v,r,q}\,h_{\ell r}^{(j)\top}M_{\ell v}G_{K_t}\xi_{\ell q}^{(j)}
\le\lambda_\ell^{(j)},\\
\bar\omega_{\ell j}(K_t)
&:=\max_{v,r}\,h_{\ell r}^{(j)\top}\Omega_{K_t}(\eta_{\ell v})h_{\ell r}^{(j)}
\le\bar\omega_{\ell j},
\end{aligned}
\label{eq:DR-current-contraction}
\end{equation}
 are such that the stationary
covariance $X_t(\eta)$ in \eqref{eq:DR-Fisher-stationary-X} and the value
matrix $P_{K_t}(\eta)$ in \eqref{eq:DR-pointwise-P} satisfy
\begin{equation}
\begin{aligned}
&\|X_t(\eta)\|_2
\le\bar x_\ell(K_t)
:=\min_{j\in\mathcal J_t}
\frac{R_{\ell j}^2\,\bar\omega_{\ell j}(K_t)}{1-\lambda_{\ell j}^{\rm cl}(K_t)^2}
\\ 
& \qquad \qquad \le\bar x_\ell^{\rm cert}  
:=\max_j\frac{R_{\ell j}^2\,\bar\omega_{\ell j}}{1-(\lambda_\ell^{(j)})^2},\\ &
\|P_{K_t}(\eta)\|_2 \le\bar p_\ell^{\rm cert}
:=\max_j\frac{\bar s_{\ell j}}{r_{\ell j}^2\big(1-(\lambda_\ell^{(j)})^2\big)}.
\end{aligned}
\label{eq:DR-certificate-PX-bounds}
\end{equation}
\end{lemma}
\begin{proof}
See Appendix~\ref{app:DR-certificate-computable-bounds}.  
\end{proof}

\nw{Lemma~\ref{prop:DR-certificate-computable-bounds} supplies the
current-policy covariance bound needed next. For
$K\in\mathcal S_{\rm cert}$, let $\bar x_\ell(K)$ denote
\eqref{eq:DR-certificate-PX-bounds} with $K_t=K$. At iteration $t$, set
$K=K_t$. The certified form $\mathfrak I_K$ dominates transition Fisher
information by Lemma~\ref{lem:DR-dominating-Fisher}.}

For $\eta\in\mathcal P$, define the whitened covariance level
\begin{equation} \label{zeta1}
\xi_K(\eta):=\lambda_{\max}\!\left(
\bar W^{-1/2}X_K(\eta)\bar W^{-1/2}\right)\ge1,
\end{equation}
\nw{because $X_K(\eta)\succeq\bar W$. The next lemma gives the finite
covariance perturbation before any budget or region is imposed.}

\begin{lemma}
\label{lem:DR-covariance-distortion}
Let $K,K+D\in\mathcal S_{\rm cert}$ and $d=\vecop(D)$. For every
$\eta\in\mathcal P$, define
\begin{equation*}
\phi_K(\eta;d):=\mathfrak I_K(\eta;D)
+2\sqrt{\mathfrak I_K(\eta;D)\bigl(\xi_K(\eta)-1\bigr)}.
\end{equation*}
If $\phi_K(\eta;d)<1$, then
\begin{equation}
\frac{X_K(\eta)}{1+\phi_K(\eta;d)}
\preceq X_{K+D}(\eta)
\preceq\frac{X_K(\eta)}{1-\phi_K(\eta;d)}.
\label{eq:DR-raw-covariance-sandwich}
\end{equation}
\end{lemma}
\begin{proof}
See Appendix~\ref{proof:lema7}.
\end{proof}

\nw{We now enforce its condition uniformly in two steps: a cellwise Fisher
budget converts the desired distortion into a Fisher level, and a vertex
envelope enforces that level over the support.}

\emph{Fisher budget.}
\nw{Prescribe the allowable distortion
$0<\bar\phi_\ell<1$ on cell $\mathcal P_\ell$ and define}
\begin{equation}
\begin{aligned}
c_\ell(K)&:=\frac{\bar x_\ell(K)}{\lambda_{\min}(\bar W)}-1,\\
\Delta_\ell^{\rm u}(K)&:=\frac12
\left(\sqrt{c_\ell(K)+\bar\phi_\ell}
-\sqrt{c_\ell(K)}\right)^2.
\end{aligned}
\label{eq:DR-cellwise-Fisher-budget}
\end{equation}
\nw{Using \eqref{zeta1}, $\xi_K(\eta)-1\le c_\ell(K)$ on $\mathcal P_\ell$. Thus,
$\mathfrak I_K(\eta;D)\le2\Delta_\ell^{\rm u}(K)$ implies}
\begin{equation}
\begin{aligned}
\phi_K(\eta;d)
&\le2\Delta_\ell^{\rm u}(K)
+2\sqrt{2c_\ell(K)\Delta_\ell^{\rm u}(K)}\\
&=\bar\phi_\ell<1.
\end{aligned}
\label{eq:DR-Fisher-budget-implication}
\end{equation}
\nw{Thus, $2\Delta_\ell^{\rm u}(K)$ is the certified Fisher level that
ensures $\phi_K(\eta;d)\le\bar\phi_\ell$.}

\emph{Fisher region.}
\nw{To remove the remaining realization dependence, let
$B_{\ell v}:=B(\eta_{\ell v})$. Equation~\eqref{eq:DR-certified-Fisher-form}
and $\|X_K(\eta)\|_2\le\bar x_\ell(K)$ give}
\begin{equation*}
\mathfrak I_K(\eta;D)
\le\bigl(\bar x_\ell(K)+\|\bar V\|_2\bigr)
\|\bar W^{-1/2}B(\eta)D\|_F^2.
\end{equation*}
\nw{The right-hand side is convex in the affine $B(\eta)$, so its cellwise
maximum occurs at a retained vertex. Define}
\begin{equation}
\mathcal G_{\ell v}(K;d):=\epsilon_I\|d\|_2^2
+\bigl(\bar x_\ell(K)+\|\bar V\|_2\bigr)
\|\bar W^{-1/2}B_{\ell v}D\|_F^2,
\label{eq:DR-certified-Fisher-functional}
\end{equation}
\nw{and its normalized support-wide envelope}
\begin{equation}
\Gamma_K(d):=
\max_{\ell\in[L],\,v\in[N_\ell^\theta]}
\frac{\mathcal G_{\ell v}(K;d)}
{2\Delta_\ell^{\rm u}(K)},
\label{eq:DR-normalized-Fisher-envelope}
\end{equation}
\nw{where $\epsilon_I>0$ bounds the region. If $\Gamma_K(d)\le1$, then for
every $\eta\in\mathcal P_\ell$, one has}
\begin{equation}
\mathfrak I_K(\eta;D)
\le\max_v\mathcal G_{\ell v}(K;d)
\le2\Delta_\ell^{\rm u}(K).
\label{eq:DR-full-Fisher-envelope-bound}
\end{equation}
Thus, \eqref{eq:DR-Fisher-budget-implication} gives
$\phi_K(\eta;d)\le\bar\phi_\ell<1$. 
\nw{Therefore, $\{d:\Gamma_K(d)\le1\}$ is the Fisher region: together with
$K+D\in\mathcal S_{\rm cert}$, it enforces the prescribed covariance
distortion for every admissible realization.}

Let $c_\ell^{\rm cert}$, $\Delta_\ell^{\rm u,cert}$, and
$\Gamma^{\rm cert}$ denote the quantities in
\eqref{eq:DR-cellwise-Fisher-budget}--\eqref{eq:DR-normalized-Fisher-envelope}
with $\bar x_\ell(K)$ replaced by $\bar x_\ell^{\rm cert}$, and let
$c_W:=\lambda_{\max}(\bar W)/\lambda_{\min}(\bar W)-1$. Since
$X_K(\eta)\succeq\bar W$ and \eqref{eq:DR-certificate-PX-bounds} give
$\lambda_{\max}(\bar W)\le\bar x_\ell(K)\le\bar x_\ell^{\rm cert}$, and
$\Delta_\ell^{\rm u}$ decreases in $c_\ell$,
\begin{equation}
\begin{aligned}
\Delta_\ell^{\rm u,cert}\le\Delta_\ell^{\rm u}(K)
&\le\Delta_{\ell,W}^{\rm u}
:=\tfrac12\big(\sqrt{c_W+\bar\phi_\ell}-\sqrt{c_W}\big)^2,\\
\Gamma_K(d)&\le\Gamma^{\rm cert}(d).
\end{aligned}
\label{eq:DR-budget-monotonicity}
\end{equation}
\nw{Hence, the current-gain Fisher region contains the certificate-wide one
and is less conservative whenever the current policy contracts more strongly.}

\subsubsection{Certified upper model}
\label{subsubsec:DR-upper-model}

The exact cost change uses $X_{K+D}(\eta)$, whereas the Sylvester model freezes
$X_K(\eta)$. We now define a candidate curvature correction. 
The next theorem proves that it bounds this omitted
covariance effect. 
For $K\in\mathcal S_{\rm cert}$ and $\eta\in\mathcal P$, it relies on the following definitions
\begin{align*}
\mathcal R_K(\eta)&:=R+B(\eta)^\top P_K(\eta)B(\eta),\\
\mathcal E_K(\eta)&:=\mathcal R_K(\eta)K
+B(\eta)^\top P_K(\eta)A(\eta),\\
\Psi_K(\eta)&:=\mathcal E_K(\eta)^\top
\mathcal R_K(\eta)^{-1}\mathcal E_K(\eta),
\end{align*}
where $\mathcal E_K(\eta)$ is the realization-wise Riccati residual. Define
\begin{equation}
\begin{aligned}
\mathbf A_K(\eta)&:=X_K(\eta)\otimes\mathcal R_K(\eta),\\
\psi_K(\eta)&:=\operatorname{tr}(\Psi_K(\eta)X_K(\eta)).
\end{aligned}
\label{eq:DR-pointwise-Fisher-quantities}
\end{equation}
and set
\[
a_K(\eta;d):=d^\top\mathbf A_K(\eta)d
=\operatorname{tr}(D^\top\mathcal R_K(\eta)DX_K(\eta)).
\]
Here, \(a_K(\eta;d)\) quantifies the quadratic effect of the policy displacement \(D\), weighted by the current stationary covariance \(X_K(\eta)\) and the realization-wise LQR curvature \(\mathcal R_K(\eta)\), while $\psi_K(\eta)$
measures the Riccati-residual energy. For $\eta\in\mathcal P_\ell$ (using any
containing cell on a shared boundary), define
\begin{equation}
\begin{aligned}
\beta_K(\eta):={}&
\frac{\sqrt{\psi_K(\eta)+\epsilon_\psi}}{1-\bar\phi_\ell}\\
&\times
\Bigl(\sqrt{2\Delta_\ell^{\rm u}(K)}
+2\sqrt{\xi_K(\eta)-1+\epsilon_\xi}\Bigr),
\end{aligned}
\label{eq:DR-residual-weighted-Fisher-beta}
\end{equation}
and, with
$\mathbf M_K(\eta):=2(X_K(\eta)+\bar V)\otimes\mathcal R_K(\eta)$, define
\begin{equation}
\begin{aligned}
\mathbf H_K^{\rm e}(\eta)
:={}&\mathbf M_K(\eta)+2\Big[
\frac{\bar\phi_\ell}{1-\bar\phi_\ell}\mathbf A_K(\eta)\\
&+\beta_K(\eta)
\big(\varsigma\mathbf F_K(\eta)
+\varsigma^{-1}\mathbf A_K(\eta)\big)
\Big].
\end{aligned}
\label{eq:DR-e-modified-curvature}
\end{equation}
The following theorem proves that the added curvature in
\eqref{eq:DR-e-modified-curvature} absorbs the frozen-covariance error.

\begin{theorem}
\label{prop:DR-e-model-validity}
Let $K,K+D\in\mathcal S_{\rm cert}$, $d=\vecop(D)$, and
$\Gamma_K(d)\le1$. Then, for every $\ell$ and
$\eta\in\mathcal P_\ell$, with the realization-wise Sylvester model
\[
\mathfrak q_K(\eta;D):=
\langle\nabla_K\bar J(K,\theta(\eta)),D\rangle
+\frac12d^\top\mathbf M_K(\eta)d,
\]
one has
\begin{equation}
\begin{aligned}
&\big|\bar J(K+D,\theta(\eta))-\bar J(K,\theta(\eta))
-\mathfrak q_K(\eta;D)\big|\\
&\quad\le
\frac{\phi_K(\eta;d)}{1-\phi_K(\eta;d)}
\Big(a_K(\eta;d)
+2\sqrt{\psi_K(\eta)a_K(\eta;d)}\Big)\\
&\quad\le
\frac12d^\top
\big(\mathbf H_K^{\rm e}(\eta)-\mathbf M_K(\eta)\big)d .
\end{aligned}
\label{eq:DR-Fisher-relative-model-gap}
\end{equation}
\end{theorem}
\begin{proof}
See Appendix~\ref{app:e-model-validity}.
\end{proof}

The two inequalities separate covariance control from descent certification.
The Fisher form
\(
\mathfrak I_K(\eta;D)=d^\top\mathbf F_K(\eta)d
\)
measures the sensitivity of the closed-loop transition law to the policy
perturbation \(D\). Through
\(
\phi_K
=
\sqrt{\mathfrak I_K}
\bigl(\sqrt{\mathfrak I_K}+2\sqrt{\xi_K-1}\bigr),
\)
the Fisher budget \(\phi_K\leq\bar\phi_\ell<1\) bounds the relative change
from \(X_K(\eta)\) to \(X_{K+D}(\eta)\). The matrix
\(
\mathbf A_K(\eta)=X_K(\eta)\otimes R_K(\eta)
\)
then translates this covariance change into a cost error:
\(
a_K(\eta;d)=d^\top\mathbf A_K(\eta) d
\)
is the direct quadratic contribution, while
\(
2\sqrt{\psi_K(\eta) a_K(\eta;d)}
\)
captures its interaction with the Riccati residual. The covariance sandwich
yields the \(\mathbf A_K(\eta)\) correction, whereas the Fisher budget and Young's
inequality bound the residual interaction by
\(
\beta_K(\eta)(\varsigma\mathbf F_K(\eta)+\varsigma^{-1}\mathbf A_K(\eta)).
\)
Thus, \(\mathbf F_K(\eta)\) controls the policy-induced covariance variation, while
\(\mathbf A_K(\eta)\) quantifies its effect on the cost. Adding these corrections
makes \(\mathbf H_K^{\rm e}(\eta)\) an upper curvature for the omitted covariance
change, so any certified Fisher-feasible step with a negative e-model
guarantees descent of the actual objective.

To lift the realization-wise bound to the domain-randomized objectives, we
average the e-curvature over the design distribution and its fixed batch to get
\(
\mathbf H_K^{\rm e}
:=
\mathbb E_{\eta\sim\varrho}
\!\left[\mathbf H_K^{\rm e}(\eta)\right],
\,
\mathbf H_{N,K}^{\rm e}
:=
\frac{1}{N}\sum_{i=1}^{N}\mathbf H_K^{\rm e}(\eta_i).
\)
The corresponding fixed-batch e-model is
\begin{equation}
\mathfrak m_{N,K}^{\rm e}(d):=
\mathbf g_N(K)^\top d+\frac12d^\top\mathbf H_{N,K}^{\rm e}d.
\label{eq:DR-SAA-e-model}
\end{equation}
This residual-weighted correction is termed e-modified, following robust
adaptive control~\cite{emod}. 

\nw{For implementation, the finite Bellman sweeps in
\eqref{eq:DR-GPI-tracked-PX} provide
\(\widetilde P_{\ell i,t}\) and \(\widetilde X_{\ell i,t}\).
Replacing \(P_{K_t}(\eta_i)\) and \(X_{K_t}(\eta_i)\) by these tracked
matrices in each sampled curvature
\(\mathbf H_{K_t}^{\rm e}(\eta_i)\), and then averaging over the batch,
gives the computable curvature \(\widetilde{\mathbf H}_t^{\rm e}\).
The smooth e-model generates the direction, while its buffered version
certifies the step}
\begin{equation}
\begin{aligned}
&\mathbf H_t^{\rm e}:=
\widetilde{\mathbf H}_t^{\rm e}+\epsilon_t^H I,\qquad
\widetilde{\mathfrak m}_t(d):=
\widetilde{\mathbf g}_t^\top d+\tfrac12d^\top\mathbf H_t^{\rm e}d,\\
&\mathcal U_t(d):=
\widetilde{\mathfrak m}_t(d)+\epsilon_t^g\|d\|_2.
\end{aligned}
\label{eq:DR-Fisher-direction-model}
\end{equation}
\nw{The two buffers cover partial evaluation. Since
$\widetilde{\mathbf H}_t^{\rm e}\succeq\widetilde{\mathbf M}_t
\succeq2\sigma_0I$, $\mathcal U_t$ is $2\sigma_0$-strongly convex and
$\mathcal U_t(0)=0$; under exact evaluation both buffers vanish.}

\begin{lemma}
\label{lem:DR-SAA-evaluation-error}
For $\eta_{\ell i}\in\mathcal P_\ell$, the sweeps
\eqref{eq:DR-GPI-P-sweep}--\eqref{eq:DR-GPI-X-sweep} satisfy
\begin{align*}
\|\widetilde P_{\ell i,t}-P_{K_t}(\eta_{\ell i})\|_2
&\le(\bar s+\bar p_\ell^{\rm cert})
\kappa_\ell^2\lambda_\ell^{2m_t},\\
\|\widetilde X_{\ell i,t}-X_{K_t}(\eta_{\ell i})\|_2
&\le(\bar\omega_\ell+\bar x_\ell^{\rm cert})
\kappa_\ell^2\lambda_\ell^{2m_t}.
\end{align*}
Let $C_g^e$ and $C_H^e$ be Lipschitz constants of the realization-wise
gradient and e-curvature contributions with respect to $(P,X)$ on the
certified set, and define
\begin{equation*}
\epsilon_t^{g}:=
C_g^e\max_\ell c_\ell^{\rm ev}\kappa_\ell^2\lambda_\ell^{2m_t},
\qquad
\epsilon_t^{H}:=
C_H^e\max_\ell c_\ell^{\rm ev}\kappa_\ell^2\lambda_\ell^{2m_t},
\end{equation*}
with
$c_\ell^{\rm ev}:=
\bar s+\bar p_\ell^{\rm cert}
+\bar\omega_\ell+\bar x_\ell^{\rm cert}$.
Then,
\begin{equation}
\|\widetilde{\mathbf g}_t-\mathbf g_N(K_t)\|_2\le\epsilon_t^g,
\qquad
\|\widetilde{\mathbf H}_t^{\rm e}
-\mathbf H_{N,K_t}^{\rm e}\|_2\le\epsilon_t^H.
\label{eq:DR-SAA-evaluation-error-bounds}
\end{equation}
\end{lemma}

\begin{proof}
Subtracting $P_{K_t}=S^{(t)}+F^\top P_{K_t}F$ from
\eqref{eq:DR-GPI-P-sweep} gives
$P^{(j)}-P_{K_t}=(F^\top)^j(P^{(0)}-P_{K_t})F^j$, where
$F=F_{\ell i,t}$. By \eqref{eq:DR-certificate-power-bound} and
\eqref{eq:DR-certificate-PX-bounds},
\[
\|P^{(m_t)}-P_{K_t}\|_2
\le\kappa_\ell^2\lambda_\ell^{2m_t}
(\bar s+\bar p_\ell^{\rm cert}).
\]
The covariance sweep is identical with $F$ in place of $F^\top$ and
$\|\Omega_{\ell i,t}\|_2\le\bar\omega_\ell$. The Lipschitz constants
transfer these realization-wise errors to each summand of
$\widetilde{\mathbf g}_t$ and $\widetilde{\mathbf H}_t^{\rm e}$, and
averaging does not increase the maximum.
\end{proof}

\nw{The compact certified set makes $C_g^e,C_H^e$ finite and offline
bounded. Combining the theorem and lemma gives the following computable SAA
upper model.}

\begin{proposition}
\label{prop:DR-certified-upper-model}
Let $\mathcal S_{\rm cert}$ and $\Gamma_{K_t}(d)$ be defined in
\eqref{Scert} and \eqref{eq:DR-normalized-Fisher-envelope}, respectively.
Then, for every $d$ with
$K_t+\operatorname{mat}(d)\in\mathcal S_{\rm cert}$ and
$\Gamma_{K_t}(d)\le1$,
\begin{equation}
f_N(K_t+\operatorname{mat}(d))-f_N(K_t)
\le\mathfrak m_{N,K_t}^{\rm e}(d)\le\mathcal U_t(d).
\label{eq:DR-SAA-majorization}
\end{equation}
where $\mathcal U_t(d)$ is defined in \eqref{eq:DR-Fisher-direction-model}. Moreover, for every $d$,
\begin{equation}
0\le\mathcal U_t(d)-\mathfrak m_{N,K_t}^{\rm e}(d)
\le2\epsilon_t^g\|d\|_2+\epsilon_t^H\|d\|_2^2.
\label{eq:DR-computed-e-model-error}
\end{equation}
\end{proposition}
\begin{proof}
See Appendix~\ref{proof:prop3}.
\end{proof}

\nw{Thus any Fisher-valid $d$ with $\mathcal U_t(d)<0$ decreases $f_N$
without a cost evaluation.}

\subsubsection{Check-and-correct Sylvester step}
\label{subsubsec:DR-steps}
\nw{The pieces now align: $\widetilde{\mathfrak m}_t$ supplies a Sylvester
direction, the controller polytope corrects it for full-support stability, and
$\mathcal U_t$ gives a Fisher-certified descent length. Only afterward may a
scalar search extend that step along the same stability-certified ray.}

\emph{Sylvester step and stability check.}
\nw{The unconstrained e-model minimizer is}
\begin{equation}
d_t^{\rm unc}:=-\left(\mathbf H_t^{\rm e}\right)^{-1}
\widetilde{\mathbf g}_t.
\label{eq:DR-Fisher-direction-proxy}
\end{equation}
\nw{Among the nonempty set
$\mathcal J_t:=\{j:k_t\in\widehat{\mathcal C}_{j,t}\}$, select the bundle
allowing the longest move along $d_t^{\rm unc}$:}
\begin{equation}
j_t^\star\in\arg\max_{j\in\mathcal J_t}
\widehat\alpha_{j,t}(d_t^{\rm unc}),
\label{eq:DR-bundle-selection}
\end{equation}
\nw{Use the previous bundle to break ties. If
\eqref{eq:DR-inner-enrichment-trigger} holds, enrich by
\eqref{eq:DR-inner-enrichment}, update the margin, and reselect. Set
$\widehat{\mathcal C}_t:=\widehat{\mathcal C}_{j_t^\star,t}$ and freeze the
selected bundle after $J_{\rm sw}$ actual switches. If
$k_t+d_t^{\rm unc}\in\widehat{\mathcal C}_t$, set
$d_t=d_t^{\rm unc}$. Otherwise solve the same e-model over the certified
polytope}
\begin{equation}
d_t:=\arg\min_{d:\,k_t+d\in\widehat{\mathcal C}_t}
\left\{
\widetilde{\mathbf g}_t^\top d+
\frac12d^\top\mathbf H_t^{\rm e}d
\right\}.
\label{eq:DR-Fisher-direction-predictor}
\end{equation}
\nw{Strict convexity makes the QP unnecessary when the unconstrained endpoint
is feasible. Fisher already shapes $d_t$ through $\mathbf H_t^{\rm e}$, but its
region is checked only after this stability correction.}

\emph{Fisher check and certified step length.}
\nw{Along the stable ray
$K_t+\alpha\operatorname{mat}(d_t)$, the upper model becomes}
\begin{equation}
\begin{aligned}
\mathcal U_t(\alpha d_t)&=-\alpha s_t+\frac12\alpha^2h_t,\\
s_t&:=-\widetilde{\mathbf g}_t^\top d_t-\epsilon_t^g\|d_t\|_2,\qquad
h_t:=d_t^\top\mathbf H_t^{\rm e}d_t .
\end{aligned}
\label{eq:DR-ray-model}
\end{equation}
\nw{Homogeneity gives
$\Gamma_{K_t}(\alpha d_t)=\alpha^2\Gamma_{K_t}(d_t)$, so the Fisher-valid
range is $\alpha\le\alpha_t^{\rm F}:=\Gamma_{K_t}(d_t)^{-1/2}$.
Minimizing \eqref{eq:DR-ray-model} over this range and the unit step gives}
\begin{equation}
\begin{aligned}
\alpha_t^{\rm c}&:=
\max\big\{0,\min\{1,\alpha_t^{\rm F},s_t/h_t\}\big\},\\
\Delta_t^{\rm c}&:=-\mathcal U_t(\alpha_t^{\rm c}d_t)
=\alpha_t^{\rm c}s_t-\frac12(\alpha_t^{\rm c})^2h_t .
\end{aligned}
\label{eq:DR-certified-step-length}
\end{equation}
\nw{No cost evaluation is needed. If $\alpha_t^{\rm F}\ge1$, Fisher does not
truncate the unit step; under exact evaluation and an unconstrained step,
$\alpha_t^{\rm c}=1$. Stop when $d_t=0$.}

\emph{Correction by safeguarded line search.}
\nw{When $\alpha_t^{\rm F}<1$, Fisher may truncate the step before the
stability boundary. The same ray remains certified up to}
\[
\alpha_t^{\max}:=
\max\{\alpha\ge0:k_t+\alpha d_t\in\widehat{\mathcal C}_t\}\ge1,
\]
\nw{which is a one-variable LP. Let
$\varphi_t(\alpha):=f_N(K_t+\alpha\operatorname{mat}(d_t))$ and, at the first
$t_\tau$ with $\Delta_{t_\tau}^{\rm c}>0$, set}
\begin{equation}
\tau_{\rm ls}:=\gamma_{\rm ls}\Delta_{t_\tau}^{\rm c},
\qquad 0<\gamma_{\rm ls}<1.
\label{eq:DR-LS-threshold}
\end{equation}
\nw{If no positive $\Delta_t^{\rm c}$ occurs, no search is used. Otherwise,
when $\alpha_t^{\rm F}<1$ and $\Delta_t^{\rm c}\ge\tau_{\rm ls}$, use at most
$n_{\rm ls}$ evaluations of any bounded scalar method on
$[\alpha_t^{\rm c},\alpha_t^{\max}]$, include $\alpha_t^{\rm c}$, and return
the best trial $\alpha_t$; in all other cases set
$\alpha_t=\alpha_t^{\rm c}$. Update}
\begin{equation}
K_{t+1}=K_t+\alpha_t\operatorname{mat}(d_t).
\label{eq:DR-ray-update}
\end{equation}
\nw{Every trial remains in the controller polytope, and inclusion of
$\alpha_t^{\rm c}$ preserves its certified decrease. The threshold makes the
number of search iterations finite.}

\begin{proposition}
\label{prop:DR-certified-descent}
Let $K_{t+1}$ be generated by the check-and-correct
policy update \eqref{eq:DR-Fisher-direction-proxy}--\eqref{eq:DR-ray-update}. Then:
\begin{enumerate}
\item[(a)] $s_t\ge2\sigma_0\|d_t\|_2^2-\epsilon_t^g\|d_t\|_2$ and
$\Delta_t^{\rm c}\ge\frac12\alpha_t^{\rm c}s_t\ge0$;
\item[(b)] $\vecop(K_{t+1})\in\widehat{\mathcal C}_t$, so $K_{t+1}$ is
contractive, hence stabilizing, for every $\eta\in\mathcal P$;
\item[(c)] $f_N(K_{t+1})-f_N(K_t)\le-\Delta_t^{\rm c}$;
\item[(d)] the line search is used in at most
$f_N(K_{t_\tau})/\tau_{\rm ls}$ iterations, and only when the Fisher check
fails.
\end{enumerate}
\end{proposition}

\begin{proof}
The feasible set of \eqref{eq:DR-Fisher-direction-predictor} is convex and
contains $0$ because $k_t\in\widehat{\mathcal C}_t$. First-order optimality at
$d_t$ in the direction $0-d_t$ gives
$-\widetilde{\mathbf g}_t^\top d_t
\ge d_t^\top\mathbf H_t^{\rm e}d_t
\ge2\sigma_0\|d_t\|_2^2$, which gives the bound on $s_t$. If $s_t\le0$, then
$\alpha_t^{\rm c}=0=\Delta_t^{\rm c}$. Otherwise
$\alpha_t^{\rm c}\le s_t/h_t$ gives
$(\alpha_t^{\rm c})^2h_t\le\alpha_t^{\rm c}s_t$, hence
$\Delta_t^{\rm c}\ge\frac12\alpha_t^{\rm c}s_t$, which proves part (a).
The inner set $\widehat{\mathcal C}_t$ is convex and contains $k_t$ because
$j_t^\star\in\mathcal J_t$. By the definition of $\alpha_t^{\max}$, it contains
$k_t+\alpha d_t$ for all $\alpha\in[0,\alpha_t^{\max}]$, and
$\alpha_t\le\alpha_t^{\max}$. Since
$\widehat{\mathcal C}_t\subseteq\mathcal C_{j_t^\star}$, contractive invariance
holds on every cell, proving part (b). Since
$\alpha_t^{\rm c}\le\min\{\alpha_t^{\rm F},\alpha_t^{\max}\}$, the
displacement $\alpha_t^{\rm c}d_t$ is certificate-feasible and satisfies
$\Gamma_{K_t}(\alpha_t^{\rm c}d_t)\le1$. Defining
$\varphi_t(\alpha):=f_N(K_t+\alpha\operatorname{mat}(d_t))$,
Proposition~\ref{prop:DR-certified-upper-model}  and \eqref{eq:DR-certified-step-length} gives
\[
\varphi_t(\alpha_t^{\rm c})-\varphi_t(0)
\le\mathcal U_t(\alpha_t^{\rm c}d_t)
=-\Delta_t^{\rm c}.
\]
The search returns the best evaluated point, and $\alpha_t^{\rm c}$ is among
them, so
$\varphi_t(\alpha_t)\le\varphi_t(\alpha_t^{\rm c})$, proving part (c).
The search runs only when $\alpha_t^{\rm F}<1$ and
$\Delta_t^{\rm c}\ge\tau_{\rm ls}$, so each such iteration decreases $f_N$
by at least $\tau_{\rm ls}$ by (c). By (c), $f_N$ never increases, and
$f_N\ge0$, proving part (d).
\end{proof}


\section{Convergence and Guarantees}
\label{sec:DR-convergence}

This section proceeds in four steps. 
Section~\ref{subsec:DR-population-improvement} gives the exact performance
difference and decomposes the Sylvester model into decrement and dispersion
terms. Sections~\ref{subsec:DR-SAA-stationarity}--\ref{subsec:DR-SAA-rate}
establish SAA stationarity, exact-evaluation convergence, and an eventual
linear rate with a dispersion and measurement-noise floor.
Sections~\ref{subsec:DR-uniform-SAA}--\ref{subsec:DR-true-cost} transfer these
results to the population surrogate, Problem~1, and the true plant.
Throughout, $\mathbb E_\varrho$ denotes expectation over
$\eta\sim\mathbb P_{\rm des}$.


For the sample-wise Sylvester operator \eqref{eq:DR-GPI-M-operator}, define its population counterpart
\begin{equation*}
\mathcal M_K[D]:=2\mathbb E_{\varrho}
\!\left[\mathcal R_K(\eta)D(X_K(\eta)+\bar V)\right],
\end{equation*}
{with matrix representation $\vecop(\mathcal M_K[D])=\mathbf M_K\vecop(D)$.}
Define also
\begin{equation*}
f(K):=\bar J_{\rm DR}(K),\, k:=\vecop(K), \,
\mathbf g(K):=\vecop(\nabla_K f(K)).
\end{equation*}
{With $\mathcal R_K$, $\mathcal E_K$, and $\Psi_K$ as in
Section~\ref{subsec:DR-Fisher-partial-improvement}, define the realization-wise
Hewer displacement and the population Sylvester step}
\begin{equation*}
D_K^{\rm H}(\eta):=-\mathcal R_K(\eta)^{-1}\mathcal E_K(\eta), \quad D_K^\star:=-{\operatorname{mat}(}\mathbf M_K^{-1}\mathbf g(K)).
\end{equation*}

After the finite bundle-switching and inner-enrichment budgets have been
exhausted, let $\widehat{\mathcal C}$ denote the fixed selected inner controller
set and define $\mathcal S:=\{K:\vecop(K)\in\widehat{\mathcal C}\}$. This set is
compact and uniformly stabilizing {, so the optimal values
$f_{\mathcal S}^*:=\min_{K\in\mathcal S}f(K)$ and
$f_{N,\mathcal S}^*:=\min_{K\in\mathcal S}f_N(K)$ are attained}. On $\mathcal S$, the Sylvester curvature
satisfies $2\sigma_0I\preceq\mathbf M_K$. {Since $\psi_K$ and $\xi_K$ are
bounded on the compact certified set and $\Delta_\ell^{\rm u}(K)\le\Delta_{\ell,W}^{\rm u}$
by \eqref{eq:DR-budget-monotonicity}, the coefficient $\beta_K$ is bounded, so
there is $\bar h_0<\infty$ with
$\|\mathbf H_K^{\rm e}(\eta)\|_2\le\bar h_0$ for all $K\in\mathcal S_{\rm cert}$
and $\eta\in\mathcal P$. The same bound holds for the population and batch
averages.}
For $K\in\mathcal S$, define the certificate-feasible direction set and the
population e-modified direction model
\begin{equation}
\begin{aligned}
\mathcal D_K^{\rm C}&:=\{d:k+d\in\widehat{\mathcal C}\},\\
\mathfrak m_K^{\rm e}(d)&:=\mathbf g(K)^\top d+\frac12d^\top\mathbf H_K^{\rm e}d.
\end{aligned}
\label{eq:DR-population-e-model}
\end{equation}
Here, $\mathbf H_K^{\rm e}$ is the population average of
\eqref{eq:DR-e-modified-curvature}. $\mathfrak m_K^{\rm e}$ is the population
counterpart of the batch e-model \eqref{eq:DR-SAA-e-model}, and
Theorem~\ref{prop:DR-e-model-validity} gives
$f(K+D)-f(K)\le\mathfrak m_K^{\rm e}(d)$ whenever {$K+D\in\mathcal S_{\rm cert}$
and} $\Gamma_K(d)\le1$.

Besides the bounds of Lemma~\ref{prop:DR-certificate-computable-bounds},
the error analysis uses
\(
\bar k_F:=\max_{k\in\mathcal K_0}\|\operatorname{mat}(k)\|_F,\,
\bar a_\ell:=\max_v\|A(\eta_{\ell v})\|_F.
\)
For later use, set
\(
\bar x:=\max_\ell\bar x_\ell^{\rm cert},
\bar p:=\max_\ell\bar p_\ell^{\rm cert},
\bar a:=\max_\ell\bar a_\ell,
\bar b:=\max_\ell\bar b_\ell,
\bar r:=\|R\|_2+\bar b^2\bar p,
\bar e:=\bar r\bar k_F+\bar b\bar p\bar a.
\)
Then, $\|\mathcal R_K(\eta)\|_2\le\bar r$ and
$\|\mathcal E_K(\eta)\|_F\le\bar e$ on the complete certified support.
The certificate-wide $\bar x$ enters only {uniform constants, such as
$r_{\rm F}$ below and $\bar\xi$ in Section~\ref{subsec:DR-SAA-rate}}. All
Fisher budgets use the current-policy bound $\bar x_\ell(K)$.
{The Fisher region contains a ball of fixed radius. Since
$\|\bar W^{-1/2}B_{\ell v}D\|_F\le\|\bar W^{-1/2}B_{\ell v}\|_2\|d\|_2$ and
$\bar x_\ell(K)\le\bar x_\ell^{\rm cert}$, \eqref{eq:DR-certified-Fisher-functional}
and \eqref{eq:DR-budget-monotonicity} give $\Gamma_K(d)\le\|d\|_2^2/r_{\rm F}^2$
for every $K\in\mathcal S_{\rm cert}$, with
\begin{equation}
r_{\rm F}^2:=\min_{\ell,v}
\frac{2\Delta_\ell^{\rm u,cert}}
{\epsilon_I+(\bar x_\ell^{\rm cert}+\|\bar V\|_2)\|\bar W^{-1/2}B_{\ell v}\|_2^2}>0.
\label{eq:DR-Fisher-radius}
\end{equation}
Hence, $\alpha_t^{\rm F}\ge r_{\rm F}/\|d_t\|_2$ along the iteration.}

\subsection{Population improvement and dispersion}
\label{subsec:DR-population-improvement}

We first state the exact population performance difference and then
characterize the corresponding frozen quadratic through realization-wise
dispersion. 

\begin{lemma}
\label{lem:DR-two-channel-PDL}
For certified $K,K'=K+\Delta$, one has
\begin{equation}
\begin{aligned}
f(K')-f(K)
&=\mathbb E_{\varrho}\!\left[
\operatorname{tr}\!\big(\Phi_K(\eta;\Delta)X_{K'}(\eta)\big)\right.\\
&\hspace{18mm}\left.
+\operatorname{tr}\!\big(\Xi_K(\eta;\Delta)\bar V\big)
\right].
\end{aligned}
\label{eq:DR-two-channel-PDL}
\end{equation}
where
\begin{align*}
\Phi_K(\eta;\Delta)
&:=(\Delta-D_K^{\rm H}(\eta))^\top\mathcal R_K(\eta)
 (\Delta-D_K^{\rm H}(\eta))-\Psi_K(\eta),\\
\Xi_K(\eta;\Delta)
&:=(\Delta+K)^\top\mathcal R_K(\eta)(\Delta+K)
-K^\top\mathcal R_K(\eta)K.
\end{align*}
\end{lemma}
\begin{proof}
See Appendix~\ref{app:DR-two-channel-PDL}.    
\end{proof}

Define the improvement potential and channel-wise dispersion by
\begin{align*}
\mathfrak n(K)
&:=\mathbb E_{\eta\sim\varrho}\!\left[
\operatorname{tr}\!\big(\Psi_K(\eta)X_K(\eta)\big)
+\operatorname{tr}\!\big(K^\top\mathcal R_K(\eta)K\bar V\big)
\right],\\
\mathfrak h(K;D)
&:=\mathbb E_{\eta\sim\varrho}\!\Big[
\Big\|\mathcal R_K(\eta)^{1/2}
(D-D_K^{\rm H}(\eta))X_K(\eta)^{1/2}\Big\|_F^2
\Big]\\
&\quad+\mathbb E_{\eta\sim\varrho}\!\Big[
\Big\|\mathcal R_K(\eta)^{1/2}(D+K)\bar V^{1/2}\Big\|_F^2
\Big],
\end{align*}
and let $\mathfrak h^\star(K):=\min_D\mathfrak h(K;D)$.

\begin{lemma}
\label{lem:DR-dispersion-identities}
For every certified $K$ and every $D$, with $d=\vecop(D)$,
$d_K^\star:=\vecop(D_K^\star)$, and
$d_K^{\rm e}:=-(\mathbf H_K^{\rm e})^{-1}\mathbf g(K)$:

\noindent (a) the Sylvester model satisfies
\begin{align}
& \mathfrak q_K(D):=\langle\nabla f(K),D\rangle 
+\tfrac12\langle D,\mathcal M_K[D]\rangle \nonumber \\ &
=\mathfrak h(K;D)-\mathfrak n(K).
\label{gK}
\end{align}
\noindent (b) the dispersion and the improvement potential satisfy
\begin{align}
&
\mathfrak n(K)
=\delta^\star(K)+\mathfrak h^\star(K),\,
\delta^\star(K):=\tfrac12\mathbf g(K)^\top\mathbf M_K^{-1}\mathbf g(K).
\label{eq:DR-unregularized-decrement} \\ &\mathfrak h(K;D)
=\mathfrak h^\star(K)
+\tfrac12(d-d_K^\star)^\top\mathbf M_K(d-d_K^\star),
\label{eq:DR-dispersion-completion}
\end{align}
 \noindent (c) $\mathfrak h^\star(K)=0$ if and only if
$D_K^{\rm H}(\eta)=D_K^\star$ for $\varrho$-almost every $\eta$ and
$(D_K^\star+K)\bar V=0$;

\noindent (d) the e-model \eqref{eq:DR-population-e-model} satisfies
\begin{equation}
\mathfrak m_K^{\rm e}(d)=-\delta_{\rm e}(K)
+\tfrac12(d-d_K^{\rm e})^\top\mathbf H_K^{\rm e}(d-d_K^{\rm e}),
\label{eq:DR-e-model-completion}
\end{equation}
with $\delta_{\rm e}(K):=\frac12\mathbf g(K)^\top(\mathbf H_K^{\rm e})^{-1}\mathbf g(K)
\in[0,\delta^\star(K)]$. If, moreover,
$\mathbf H_K^{\rm e}\preceq\kappa_{\rm e}\mathbf M_K$, then
$\delta_{\rm e}(K)\ge\delta^\star(K)/\kappa_{\rm e}$.
\end{lemma}
\begin{proof}
See Appendix~\ref{app:DR-dispersion-identities}.    
\end{proof}

\subsection{SAA stationarity and whole-sequence convergence}
\label{subsec:DR-SAA-stationarity}
\label{subsec:DR-SAA-convergence}
The fixed-SAA analysis rests on two facts from Section~\ref{subsubsec:DR-steps}:
every iteration keeps the gain certified and decreases $f_N$ by at least
$\Delta_t^{\rm c}$ (Proposition~\ref{prop:DR-certified-descent}). {Let
$\bar D:=\max\{\|k-k'\|_2:k,k'\in\bigcup_j\mathcal C_j\}$, so that
$\|d_t\|_2\le\bar D$, and let $\bar h:=\bar h_0+2\sup_t\epsilon_t^H$, so that
$\|\mathbf H_t^{\rm e}\|_2\le\bar h$ by Lemma~\ref{lem:DR-SAA-evaluation-error}.}
By Proposition~\ref{prop:DR-certified-descent}(d), the line search is used in
finitely many iterations. Let $T_{\rm fix}$ be an index after which the
certificate $\widehat{\mathcal C}$ is fixed and the search is no longer used.
{The accuracy of partial evaluation is measured by the radii
$\epsilon_t^g,\epsilon_t^H$ of Lemma~\ref{lem:DR-SAA-evaluation-error}, which
bound the gradient and curvature errors in
\eqref{eq:DR-SAA-evaluation-error-bounds}. We call the evaluation exact when
$\epsilon_t^g=\epsilon_t^H=0$, that is, when the Lyapunov equations
\eqref{eq:DR-GPI-P-sweep}--\eqref{eq:DR-GPI-X-sweep} are solved to their fixed
points.}

\begin{theorem}
\label{thm:DR-SAA-convergence}
Suppose {the radii in \eqref{eq:DR-SAA-evaluation-error-bounds} satisfy}
$\epsilon_t^g\to0$ and $\sup_t\epsilon_t^H<\infty$. Then {, with
$\Delta_t^{\rm c}$ and $d_t$ as in \eqref{eq:DR-certified-step-length} and
\eqref{eq:DR-Fisher-direction-predictor},}
$\sum_t\Delta_t^{\rm c}<\infty$, $d_t\to0$, and every accumulation point
$K_\star$ of the iterates is constrained stationary for the fixed SAA problem
\begin{equation}
0\in
\mathbf g_N(K_\star)+
\mathsf N_{\mathcal S}(k_\star),
\qquad k_\star=\vecop(K_\star).
\label{eq:DR-SAA-accumulation-stationarity}
\end{equation}
Moreover, the Fisher truncation is inactive for all large $t$:
$\alpha_t^{\rm F}\ge1$, with $\alpha_t^{\rm F}$ as in
Section~\ref{subsubsec:DR-steps} {, so the Fisher check passes and the line
search is not used after finitely many iterations}. If, in addition, the fixed-batch evaluation is exact
after $T_{\rm fix}$, then {$\alpha_t=1$ for all large $t$}, so
$K_{t+1}=K_t+\operatorname{mat}(d_t)$, with $d_t$
the certified Sylvester step \eqref{eq:DR-Fisher-direction-predictor}. The
whole sequence converges to a constrained stationary point of $f_N$.
\end{theorem}
\begin{proof} See Appendix~\ref{app:DR-SAA-convergence}. \end{proof}

\begin{remark}
\label{rem:DR-evaluation-conditions}
{
The
first part of the theorem, stationarity of accumulation points, places no
restriction on the dispersion of the realization-wise Riccati targets.}
{Exact evaluation is needed for the unit step: for a nonzero unconstrained
step, $s_t=h_t-\epsilon_t^g\|d_t\|_2<h_t$ whenever $\epsilon_t^g>0$, so the
certified length stays below one under partial evaluation.}
\end{remark}

\subsection{{Linear rate with a dispersion floor}}
\label{subsec:DR-SAA-rate}
The rate is obtained in two steps: a bound on the SAA gap by the Sylvester
decrement and the dispersion, and a lower bound on the certified decrease by
the Sylvester decrement. Let $\mathfrak n_N$, $\mathfrak h_N$,
$\mathfrak h_N^\star$, and $\delta_N^\star$ denote the quantities of
Section~\ref{subsec:DR-population-improvement} with $\mathbb E_\varrho$ replaced
by the batch average, and let
$\mathbf M_{N,K}:=\frac1N\sum_{i=1}^N\mathbf M_K(\eta_i)$. Lemmas~\ref{lem:DR-two-channel-PDL} and~\ref{lem:DR-dispersion-identities}
hold for any probability law on $\mathcal P$ in place of $\mathbb P_{\rm des}$,
so they hold for the empirical law of the batch; in particular,
$\mathfrak n_N=\delta_N^\star+\mathfrak h_N^\star$ with
$\delta_N^\star(K)=\frac12\mathbf g_N(K)^\top\mathbf M_{N,K}^{-1}\mathbf g_N(K)$.
Define the covariance ratio and the curvature ratio
\begin{equation}
\bar\xi:=\frac{\bar x}{\lambda_{\min}(\bar W)},\,
\kappa_{\rm e}:=\sup_{K\in\mathcal S}\lambda_{\max}\big(\mathbf M_{N,K}^{-1/2}
\mathbf H_{N,K}^{\rm e}\mathbf M_{N,K}^{-1/2}\big).
\label{eq:DR-rate-constants}
\end{equation}
Since $X_K(\eta)\succeq\bar W$, $\mathbf H_{N,K}^{\rm e}\succeq\mathbf M_{N,K}\succeq2\sigma_0I$,
and $\|\mathbf H_{N,K}^{\rm e}\|_2\le\bar h_0$, both ratios are finite and
$\bar\xi\ge1$, $1\le\kappa_{\rm e}\le\bar h_0/(2\sigma_0)$.

\begin{proposition}
\label{prop:DR-dispersion-link}
For every $K\in\mathcal S$,
\begin{equation}
f_N(K)-f_{N,{\mathcal S}}^*\le\bar\xi\big(\delta_N^\star(K)+\mathfrak h_N^\star(K)\big).
\label{eq:DR-dispersion-gap}
\end{equation}
\end{proposition}

\begin{proof}
Let $K^\circ$ minimize $f_N$ over $\mathcal S$ and set $D:=K^\circ-K$. By
Lemma~\ref{lem:DR-two-channel-PDL}, applied for each realization,
$\bar J(K^\circ,\theta(\eta))-\bar J(K,\theta(\eta))
=\operatorname{tr}(\Phi_K(\eta;D)X_{K^\circ}(\eta))+\operatorname{tr}(\Xi_K(\eta;D)\bar V)$,
with $\Phi_K\succeq-\Psi_K$ and $\Xi_K\succeq-K^\top\mathcal R_KK$ by the completed
squares. Since $\Psi_K\succeq0$, $X_{K^\circ}(\eta)\preceq\bar xI$, and
$X_K(\eta)\succeq\bar W\succeq\lambda_{\min}(\bar W)I$, we have
$X_{K^\circ}(\eta)\preceq\bar\xi X_K(\eta)$, so, {using $\bar\xi\ge1$,}
$\bar J(K,\theta(\eta))-\bar J(K^\circ,\theta(\eta))
\le\bar\xi\big[\operatorname{tr}(\Psi_KX_K)+\operatorname{tr}(K^\top\mathcal R_KK\bar V)\big]$.
Averaging over the batch gives
$f_N(K)-f_N(K^\circ)\le\bar\xi\,\mathfrak n_N(K)=\bar\xi(\delta_N^\star+\mathfrak h_N^\star)$,
which is \eqref{eq:DR-dispersion-gap}.
\end{proof}

{Thus, the SAA gap is controlled by two terms: the decrement $\delta_N^\star$
of the common Sylvester step, which the algorithm can remove, and the
dispersion $\mathfrak h_N^\star$ of the realization-wise Riccati targets around
that step, which no common gain can remove. The next theorem shows that, once
the full Sylvester step is taken, each iteration removes a fixed fraction of
$\delta_N^\star$.}

\begin{theorem}
\label{thm:DR-SAA-linear}
{Assume exact fixed-batch evaluation after $T_{\rm fix}$, let $K_\star$ be
the limit given by Theorem~\ref{thm:DR-SAA-convergence}, and suppose that
$\vecop(K_\star)$ lies in the interior of $\widehat{\mathcal C}$. Let
\begin{equation}
\mu_N:=\frac1{\bar\xi\kappa_{\rm e}}\in(0,1],\qquad
\epsilon_{N,\rm disp}:=\bar\xi\sup_{K\in\mathcal S}\mathfrak h_N^\star(K).
\label{eq:DR-rate-and-floor}
\end{equation}
Then, there is $T\ge T_{\rm fix}$ such that the SAA gap
$e_t:=f_N(K_t)-f_{N,\mathcal S}^*$ satisfies, for all $t\ge T$,}
\begin{equation}
e_{t+1}\le(1-\mu_N)e_t+\mu_N\epsilon_{N,\rm disp}{.}
\label{eq:DR-SAA-linear-rate}
\end{equation}
{Hence, $e_t\le(1-\mu_N)^{t-T}e_T+\epsilon_{N,\rm disp}$ for $t\ge T$, and
$\limsup_te_t\le\epsilon_{N,\rm disp}$. In particular, if
$\mathfrak h_N^\star\equiv0$ on $\mathcal S$, the gap converges linearly to
zero.}
\end{theorem}

\begin{proof}
{Under exact evaluation, $\widetilde{\mathbf g}_t=\mathbf g_N(K_t)$ and
$\mathbf H_t^{\rm e}=\mathbf H_{N,K_t}^{\rm e}$. Since $K_\star$ is interior and
constrained stationary, $\mathbf g_N(K_\star)=0$. By continuity and
$\mathbf H_{N,K}^{\rm e}\succeq2\sigma_0I$, $d_t^{\rm unc}\to0$ in
\eqref{eq:DR-Fisher-direction-proxy}, so $k_t+d_t^{\rm unc}\in\widehat{\mathcal C}$
and $d_t=d_t^{\rm unc}$ for all large $t$. By
Theorem~\ref{thm:DR-SAA-convergence}, the Fisher check also passes for all
large $t$. Let $T$ be such that both hold for $t\ge T$. Then $\alpha_t^{\rm c}=1$
and $s_t=h_t=\mathbf g_N^\top(\mathbf H_{N,K_t}^{\rm e})^{-1}\mathbf g_N$ in
\eqref{eq:DR-ray-model}, so \eqref{eq:DR-certified-step-length} and
Lemma~\ref{lem:DR-dispersion-identities}(d) give
$\Delta_t^{\rm c}=\frac12\mathbf g_N^\top(\mathbf H_{N,K_t}^{\rm e})^{-1}\mathbf g_N
\ge\delta_N^\star(K_t)/\kappa_{\rm e}$. With \eqref{eq:DR-dispersion-gap},
$e_t\le\bar\xi\kappa_{\rm e}\Delta_t^{\rm c}+\epsilon_{N,\rm disp}$, that is,
$\Delta_t^{\rm c}\ge\mu_N(e_t-\epsilon_{N,\rm disp})$.
Proposition~\ref{prop:DR-certified-descent}(c) gives
$e_{t+1}\le e_t-\Delta_t^{\rm c}$, which proves \eqref{eq:DR-SAA-linear-rate}.
Iterating it gives the last claim.}
\end{proof}

\begin{remark}
\label{rem:DR-interior-condition}
Theorem~\ref{thm:DR-SAA-linear} gives an eventual linear rate under exact
evaluation when the limit is interior to $\widehat{\mathcal C}$. On the
boundary, stationarity and convergence still hold, but no rate is claimed.
The rate $\mu_N$ depends on \eqref{eq:DR-rate-constants}, while the floor
$\epsilon_{N,\rm disp}$ reflects Riccati-target dispersion and measurement
noise. For $\bar V=0$, the floor vanishes when the targets agree, regardless
of the uncertainty size.
\end{remark}

\subsection{Uniform SAA transfer}
\label{subsec:DR-uniform-SAA}
Let the fixed compact certified controller domain be
$\mathcal S_{\rm cert}:=\{K:\vecop(K)\in\bigcup_j\mathcal C_j\}$.  The final selected
set $\mathcal S$ is a subset of this fixed domain and may be chosen adaptively from
the batch.  By uniform stability and compactness, there are finite
constants such that $0\le\bar J(K,\theta(\eta))\le B_J$, each coordinate of
$\nabla_k\bar J(K,\theta(\eta))$ lies in an interval of length at most $B_g$,
and the realization-wise cost and gradient are Lipschitz in $K$ on $\mathcal S_{\rm cert}$
with constants $L_J$ and $L_g$, for all $\eta\in\mathcal P$.  If $M(r)$ is the size of an $r$-net of
$\mathcal S_{\rm cert}$, then Hoeffding's inequality at the net points, a
union bound over the net points, the gradient coordinates, and the two
events, and the Lipschitz constants give, with probability
at least $1-\delta$,
\begin{align}
& e_{J,N}(r):=B_J\sqrt{\frac{\log(4M(r)/\delta)}{2N}},\nonumber\\ &
e_{g,N}(r):=B_g\sqrt{\frac{\log(4d_KM(r)/\delta)}{2N}}, \nonumber \\ &
\sup_{K\in\mathcal S_{\rm cert}}|f_N(K)-f(K)|
\le2L_Jr+e_{J,N}(r),\nonumber\\ &
\sup_{K\in\mathcal S_{\rm cert}}\|\mathbf g_N(K)-\mathbf g(K)\|_2
\le2L_gr+\sqrt{d_K}e_{g,N}(r).
\label{eq:DR-uniform-SAA-bounds}
\end{align}
For example, if $R_{\rm cert}:=\max_{k\in\mathcal S_{\rm cert}}\|k\|_2$, then
$M(r)\le(1+2R_{\rm cert}/r)^{d_K}$.  Denote the right-hand sides after choosing
$r$ by $\epsilon_N^J$ and $\epsilon_N^g$.  Because the event is uniform over the fixed domain $\mathcal S_{\rm cert}$, it remains valid even though the final inner
certificate set is selected from the same batch. We call this the uniform
SAA event, and all statements below hold on it. In particular, for the final
$\mathcal S$,
\begin{equation}
f(K)-f_{\mathcal S}^*
\le f_N(K)-f_{N,\mathcal S}^*+2\epsilon_N^J,
\qquad K\in\mathcal S,
\label{eq:DR-population-from-SAA}
\end{equation}
because $f(K)\le f_N(K)+\epsilon_N^J$ and
$f_{\mathcal S}^*=f(K')\ge f_N(K')-\epsilon_N^J\ge f_{N,\mathcal S}^*-\epsilon_N^J$
for a minimizer $K'$ of $f$ over $\mathcal S$. Moreover, with the SAA
stationarity residual
$\varepsilon_t^{\rm sta}:=\operatorname{dist}(0,\mathbf g_N(K_t)+\mathsf N_{\mathcal S}(k_t))$,
the triangle inequality and \eqref{eq:DR-uniform-SAA-bounds} give
\begin{equation}
\operatorname{dist}\!\left(0,\mathbf g(K_t)+\mathsf N_{\mathcal S}(k_t)\right)
\le\varepsilon_t^{\rm sta}+\epsilon_N^g.
\label{eq:DR-population-stationarity-from-SAA}
\end{equation}
For every iteration, the same uniform event and
Proposition~\ref{prop:DR-certified-descent}(c) give the population
surrogate bound
\begin{equation}
f(K_{t+1})-f(K_t)
\le-\Delta_t^{\rm c}+2\epsilon_N^J.
\label{eq:DR-population-accepted-difference}
\end{equation}
Thus, the population surrogate is guaranteed to decrease whenever
$\Delta_t^{\rm c}>2\epsilon_N^J$.

\subsection{Transfer to the {Problem-1 objective and the true plant}}
\label{subsec:DR-true-cost}
{Recall the Problem-1 objective, the non-surrogate expected cost under the
design distribution,}
\(
J_{\rm DR}(K){=}
\mathbb E_{\eta\sim\mathbb P_{\rm des}}[J(K,\theta(\eta))]{,}
\)
{and reserve the term true-plant cost for $J(K,\theta^\star)$.}
By the proof of Lemma~\ref{lem:certified-pointwise-cost},
$\bar J-J=\operatorname{tr}(P_K(\bar W-W))
+\operatorname{tr}(P_KB(\eta)K(\bar V-V)K^\top B(\eta)^\top)
+\operatorname{tr}(S_K(\bar V-V))$ with $\bar W\succeq W$ and $\bar V\succeq V$.
Bounding each trace by $n$ times the product of spectral norms and using
Lemma~\ref{prop:DR-certificate-computable-bounds} gives, for $\eta\in\mathcal P_\ell$,
\begin{align}
& 0\le\bar J(K,\theta(\eta))-J(K,\theta(\eta)) \le n\bar p_\ell^{\rm cert} \nonumber\\
&
\times \left(\|\bar W-W\|_2+
\bar b_\ell^2\bar k_2^2\|\bar V-V\|_2\right)
+n\bar s\|\bar V-V\|_2.
\label{eq:DR-pointwise-true-cost-gap}
\end{align}
Set
\(
\epsilon_{\rm cert}^{J}:=\max_\ell\big[
n\bar p_\ell^{\rm cert}\big(\|\bar W-W\|_2
+\bar b_\ell^2\bar k_2^2\|\bar V-V\|_2\big)
{}+n\bar s\|\bar V-V\|_2\big].
\)
The quantity $\epsilon_{\rm cert}^{J}$ depends on the unknown true noise covariances
$W$ and $V$ and is therefore a bound.
Then, taking $\mathbb E_\varrho$ of \eqref{eq:DR-pointwise-true-cost-gap},
\begin{equation}
0\le f(K)-J_{\rm DR}(K)\le\epsilon_{\rm cert}^{J}.
\label{eq:DR-population-true-DR-gap}
\end{equation}
The certificate error is therefore paid only once in an objective-gap transfer: it
upper-bounds the difference between the population surrogate objective and the true
expected-cost objective.
Combining \eqref{eq:DR-population-true-DR-gap} with
\eqref{eq:DR-population-accepted-difference} yields, for every
iteration,
\begin{equation}
J_{\rm DR}(K_{t+1})-J_{\rm DR}(K_t)
\le-
\Delta_t^{\rm c}+2\epsilon_N^J+\epsilon_{\rm cert}^{J}.
\label{eq:DR-true-expected-accepted-difference}
\end{equation}
Hence a {Problem-1} descent certificate is available whenever
$\Delta_t^{\rm c}>2\epsilon_N^J+\epsilon_{\rm cert}^{J}$.
Let $J_{{\rm DR},\mathcal S}^{*}$ denote the minimum of {$J_{\rm DR}$} over $\mathcal S$, attained at some $K^{\rm true}$ because $\mathcal S$ is
compact and the cost is continuous on it. Then, for every $K\in\mathcal S$,
\begin{equation}
J_{\rm DR}(K)-J_{{\rm DR},\mathcal S}^{*}
\le
f_N(K)-f_{N,\mathcal S}^*+2\epsilon_N^J+\epsilon_{\rm cert}^{J}.
\label{eq:DR-true-DR-from-SAA}
\end{equation}
Indeed, by \eqref{eq:DR-population-true-DR-gap},
$J_{\rm DR}(K)\le f(K)$ and
$J_{{\rm DR},\mathcal S}^{*}=J_{\rm DR}(K^{\rm true})
\ge f(K^{\rm true})-\epsilon_{\rm cert}^{J}\ge f_{\mathcal S}^*-\epsilon_{\rm cert}^{J}$.
Subtracting and applying \eqref{eq:DR-population-from-SAA} gives
\eqref{eq:DR-true-DR-from-SAA}. Hence, the SAA rate transfers to the
{Problem-1 objective}.

\begin{corollary}
\label{cor:DR-true-expected-linear}
On the uniform SAA event, under the assumptions of
Theorem~\ref{thm:DR-SAA-linear},
\begin{equation}
\limsup_{t\to\infty}
\left[J_{\rm DR}(K_t)-J_{{\rm DR},\mathcal S}^{*}\right]
\le
\epsilon_{N,\rm disp}+2\epsilon_N^J+\epsilon_{\rm cert}^{J}.
\end{equation}
The {Problem-1} gap thus has the nonvanishing fixed-batch and
certificate floor $2\epsilon_N^J+\epsilon_{\rm cert}^{J}$.
\end{corollary}

\begin{proof}
For $t\ge T_{\rm fix}$, $K_t\in\mathcal S$ by
Proposition~\ref{prop:DR-certified-descent}(b) {, and
$e_t=f_N(K_t)-f_{N,\mathcal S}^*$}. Applying
\eqref{eq:DR-true-DR-from-SAA} at $K=K_t$ and using
$\limsup_te_t\le\epsilon_{N,\rm disp}$ from Theorem~\ref{thm:DR-SAA-linear}
gives the claim.
\end{proof}

For a particular true realization $\theta^*=\theta(\eta^*)$ with
$\eta^*\in\mathcal P$, an average SAA objective cannot
by itself control that realization.  Define the explicit coverage term
\(
\epsilon_{\rm cov}^{J}:=
\sup_{K\in\mathcal S,\,\eta\in\mathcal P}
[\bar J(K,\theta(\eta))-f(K)]_+.
\)
Then, for $t\ge T_{\rm fix}$, Lemma~\ref{lem:certified-pointwise-cost} and
the uniform SAA event give
\begin{equation}
J(K_t,\theta^*)\le \bar J(K_t,\theta^*)\le f(K_t)+\epsilon_{\rm cov}^{J}\le
f_N(K_t)+\epsilon_N^J+\epsilon_{\rm cov}^{J}.
\label{eq:DR-particular-true-cost-bound}
\end{equation}

{The guarantees combine as follows. Let $\mathcal E_{\rm I}$ be the Part~I
event, of probability at least $1-\delta_{\rm I}$, on which the true model
satisfies $\theta^\star=\theta(\eta^\star)$ for some $\eta^\star\in\mathcal P$ and
\eqref{eq:WV-certificates} holds, and let $\mathcal E_{\rm SAA}$ be the uniform
SAA event, of probability at least $1-\delta$.}

\begin{corollary}
\label{cor:DR-end-to-end}
{With probability at least $1-\delta_{\rm I}-\delta$, the following hold
simultaneously: every iterate $K_t$ stabilizes the true plant; the population
and Problem-1 descent bounds \eqref{eq:DR-population-accepted-difference} and
\eqref{eq:DR-true-expected-accepted-difference} hold at every iteration; and the
true-plant bound \eqref{eq:DR-particular-true-cost-bound} holds for
$t\ge T_{\rm fix}$.}
\end{corollary}

\begin{proof}
{On $\mathcal E_{\rm I}$, $\eta^\star\in\mathcal P$, so
Proposition~\ref{prop:DR-certified-descent}(b) gives stability of
$A_{K_t}(\eta^\star)$ for every $t$; this part holds for every batch. On
$\mathcal E_{\rm I}\cap\mathcal E_{\rm SAA}$, whose probability is at least
$1-\delta_{\rm I}-\delta$ by the union bound, the remaining bounds hold as
shown above.}
\end{proof}

\subsection{Interpretation and special cases}

\begin{corollary}
\label{cor:DR-noise-free-heterogeneous}
Suppose $\bar V=0$ {, so that the measurement-noise channel of $\mathfrak h$
vanishes}.  Then,
\begin{equation}
\mathfrak h^\star(K)
=
\mathbb E_\varrho\!\left[
\left\|
\mathcal R_K(\eta)^{1/2}
\bigl(D_K^\star-D_K^{\rm H}(\eta)\bigr)
X_K(\eta)^{1/2}
\right\|_F^2
\right].
\label{eq:DR-noise-free-dispersion-exact}
\end{equation}
If the realization-wise directions coincide almost everywhere, then
$\mathfrak h^\star(K)=0$.  More generally, if
\begin{equation*}
\sup_{K\in\mathcal S,\,\eta\in\mathcal P}
\|D_K^{\rm H}(\eta)-D_K^\star\|_F
\le\varepsilon_{\rm het},
\end{equation*}
$\mathcal R_K(\eta)\preceq\bar r I$, and
$X_K(\eta)\preceq\bar x I$, then
\begin{equation}
\sup_{K\in\mathcal S}\mathfrak h^\star(K)
\le
\bar r\bar x\,\varepsilon_{\rm het}^2.
\label{eq:DR-heterogeneity-dispersion-bound}
\end{equation}
The same bound holds for $\sup_{K\in\mathcal S}\mathfrak h_N^\star(K)$, so the
{floor in \eqref{eq:DR-rate-and-floor}} satisfies
$\epsilon_{N,\rm disp}\le\bar\xi\bar r\bar x\varepsilon_{\rm het}^2$.
\end{corollary}

\begin{proof}
Setting $\bar V=0$ in the definition of $\mathfrak h$ removes its second
channel.  The minimizer of the remaining quadratic is $D_K^\star$, which
gives \eqref{eq:DR-noise-free-dispersion-exact}.  If all
$D_K^{\rm H}(\eta)$ coincide, their common value is the unique minimizer and
the integrand vanishes.  Under the uniform bounds,
\(
\left\|
\mathcal R_K(\eta)^{1/2}
(D_K^\star-D_K^{\rm H}(\eta))
X_K(\eta)^{1/2}
\right\|_F^2\le
\bar r\bar x
\|D_K^\star-D_K^{\rm H}(\eta)\|_F^2
\le
\bar r\bar x\varepsilon_{\rm het}^2.
\)
Taking expectations and the supremum over $K$ proves
\eqref{eq:DR-heterogeneity-dispersion-bound}. For the batch, the same pointwise
bound at the batch samples gives
$\mathfrak h_N^\star(K)\le\mathfrak h_N(K;D_K^\star)\le\bar r\bar x\varepsilon_{\rm het}^2$.
\end{proof}

\begin{remark}
\label{cor:DR-Hewer}
If $\mathbb P_{\rm des}$ is a point mass, $\bar V=0$, evaluation is exact, and
the certificate constraint is inactive, the Fisher correction is not needed:
the exact advantage minimizer
\(
K_{t+1}
=
-\left(R+B^\top P_tB\right)^{-1}B^\top P_tA
\)
is the classical Newton-like update~\cite{Hewer}, which is stabilizing, and
Lemma~\ref{lem:DR-two-channel-PDL} gives
$f(K_{t+1})-f(K_t)=-\operatorname{tr}(\Psi_{K_t}X_{K_{t+1}})\le0$ because
$\Phi_{K_t}(D_{K_t}^{\rm H})=-\Psi_{K_t}$. 
\end{remark}

\section{Simulation Results}
\label{sec:simulation}

We consider a third-order system with the coupling structure of the
motivating example,
\begin{align*}
A^\star=
\begin{bmatrix}
0.4 & 0.5 & 0\\
0.5 & 0.4 & 0\\
0 & 0.3 & 0.5
\end{bmatrix},
\qquad
B^\star=
\begin{bmatrix}
0\\ 1\\ 0.5
\end{bmatrix},
\end{align*}
we assume the couplings $a_{12}$ and $a_{21}$ are
poorly known, with a prior half-width of $0.45$, and the other entries
have a prior half-width of $0.01$. 
 $\bar V=0$ and $\bar W=\Sigma^{(J)}$, the
covariance certificate of Part~I. The true process-noise covariance is
$W=0.02I_3$ and we assume $V=0$ (to compare our results with policy gradients).  Applying Part~I with a small set of data ($T=100$), returns the model set $\Theta_r^{(J)}$, in which
$a_{12},a_{21}\in[0.05,0.95]$, and the covariance certificate
$\Sigma^{(J)}$. The the learned covariance is approximately $\bar W\approx\diag(1.35,\,0.88,\,0.16)$.
The design uses $Q=I_3$, $R=3$. The design distribution $P_{\rm des}$ is a Gaussian centered at
the prior center.

\emph{Partition and certificates.} The domain is split into $L=4$ cells
of equal width along $a_{12}$, one of the two most uncertain coordinates.
The other eleven coordinates are not split. The certificates are built around a single stabilizing gain. Let
$(A_c,B_c)$ be the center of the learned set $\Theta_r^{(J)}$, and let
$K_{\rm ce}$ be the LQR gain of this center model. This gain is not
robust: its closed-loop spectral radius reaches $1.13$ at a vertex of the
domain. We therefore compute a second gain $K_{\rm rob}$ by minimizing,
with a derivative-free search started at $K_{\rm ce}$, the largest
closed-loop spectral radius over the vertices of the domain, which gives
$0.66$. The certificates are built for the gain
$K_0=0.4K_{\rm rob}+0.6K_{\rm ce}=[-0.398,\,-0.455,\,0.002]$, which is the
closest to $K_{\rm ce}$ on the segment between the two that could be
certified, and which is also the initial gain of the algorithm. For each
cell, a symmetric polytope $X_\ell=\{x:|Hx|\le1\}$ is built as follows.
Starting from the unit box, we compute exactly, by linear programs over
the cell, the largest violation of
$A_{K_0}(\theta)X_\ell\subseteq\lambda_\ell X_\ell$, and add a facet that
removes it; this is repeated until no violation is left. The rate
$\lambda_\ell$ is the smallest value on a grid of step $0.02$ for which
the procedure terminates. Table~\ref{tab:cert} reports the result. The
certificates are small in the cells where $a_{12}$ is large and the
contraction is slow, and larger in the first cell.

\begin{table}[t]
\centering
\caption{Polyhedral certificates of the bundle used by FM-DR-GPI.}
\label{tab:cert}
\begin{tabular}{@{}ccccc@{}}
\toprule
cell & range of $a_{12}$ & facets & vertices & $\lambda_\ell$\\
\midrule
$1$ & $[0.05,0.28]$ & $56$ & $100$ & $0.66$\\
$2$ & $[0.28,0.50]$ & $14$ & $24$  & $0.78$\\
$3$ & $[0.50,0.73]$ & $10$ & $16$  & $0.88$\\
$4$ & $[0.73,0.95]$ & $18$ & $32$  & $0.96$\\
\bottomrule
\end{tabular}
\end{table}

\emph{Certified gain sets.} For fixed certificates, the invariance conditions are linear in $K$ for every facet--vertex pair and every vertex of the cell. We use a bundle of five certificate families $\mathcal H^{(j)}$, $j=1,\dots,5$, each yielding one gain set $\mathcal C_j$. For the bundle of Table~I there are 6672 facet--vertex pairs. The set of gains satisfying all of them is computed exactly by cutting planes followed by removal of redundant inequalities, and it is a polytope with only 21 linear inequalities in $K\in\mathbb R^{1\times 3}$. The five gain sets have 33, 36, 28, 19 and 21 inequalities, the last being the set of Table~I. Only these are used online.

\emph{Algorithms.} FM-DR-GPI is run with a fixed batch of $N=40$ models
drawn from $P_{\rm des}$, starting from the anchor with $s=0.6$. It is
compared with the minibatch policy gradient of \cite{nguyenle2026globally}, in which
every step draws $M$ fresh models from $P_{\rm des}$ and moves along the
average of their $M$ LQR gradients with a constant step size, with no
constraint on the gain. We use $M=8$ with $3000$ steps and $M=64$ with
$1000$ steps, from the same initial gain.

\emph{Results.} FM-DR-GPI stops after $97$ iterations and returns
$K=[-0.380,\,-0.424,\,0.002]$. Two corrections act on the
Fisher-majorized step. The stability correction, which replaces the
unconstrained step by the solution of a quadratic program over
$\mathcal C_j$, is active at all $97$ iterations, because the
unconstrained minimizer lies outside the certified set. At convergence, two
of the $21$ inequalities are active. The Fisher correction, which
shortens the step to keep the closed-loop distributions close, is active
at the first $46$ iterations, and one line search is performed. Policy gradient converges without any
destabilizing step and returns $K=[-0.248,\,-0.275,\,-0.072]$ for $M=64$.

Fig.~\ref{fig:ex2}~(bottom) compares the largest closed-loop spectral
radius of the final gains over the $20000$ models sampled from
$P_{\rm des}$ and over the whole set $\Theta_r^{(J)}$. The
policy-gradient gain is stable on all samples ($0.96$) but not on the set
($1.11$). Fig.~\ref{fig:ex2}~(top) shows why: the models that it
destabilizes form a corner of the set, where $a_{12}$ and $a_{21}$ are
both large, that $P_{\rm des}$ practically never visits. Since the true
system is only known to lie in $\Theta_r^{(J)}$, this gain does not
guarantee stability of the true closed loop. The FM-DR-GPI gain is stable
on the whole set ($0.98$).

The domain-randomized cost of FM-DR-GPI is $8.2\%$ above that of policy
gradient. This is expected: policy gradient minimizes the average cost
without constraints, whereas FM-DR-GPI minimizes it over the certified
set, and the unconstrained minimizer itself has spectral radius $1.11$ on
the set. The gap is therefore the price of robust stability, which
requires a larger gain than the one that is optimal for the typical
models.

\begin{figure}[t]
\centering
\includegraphics[width=\columnwidth]{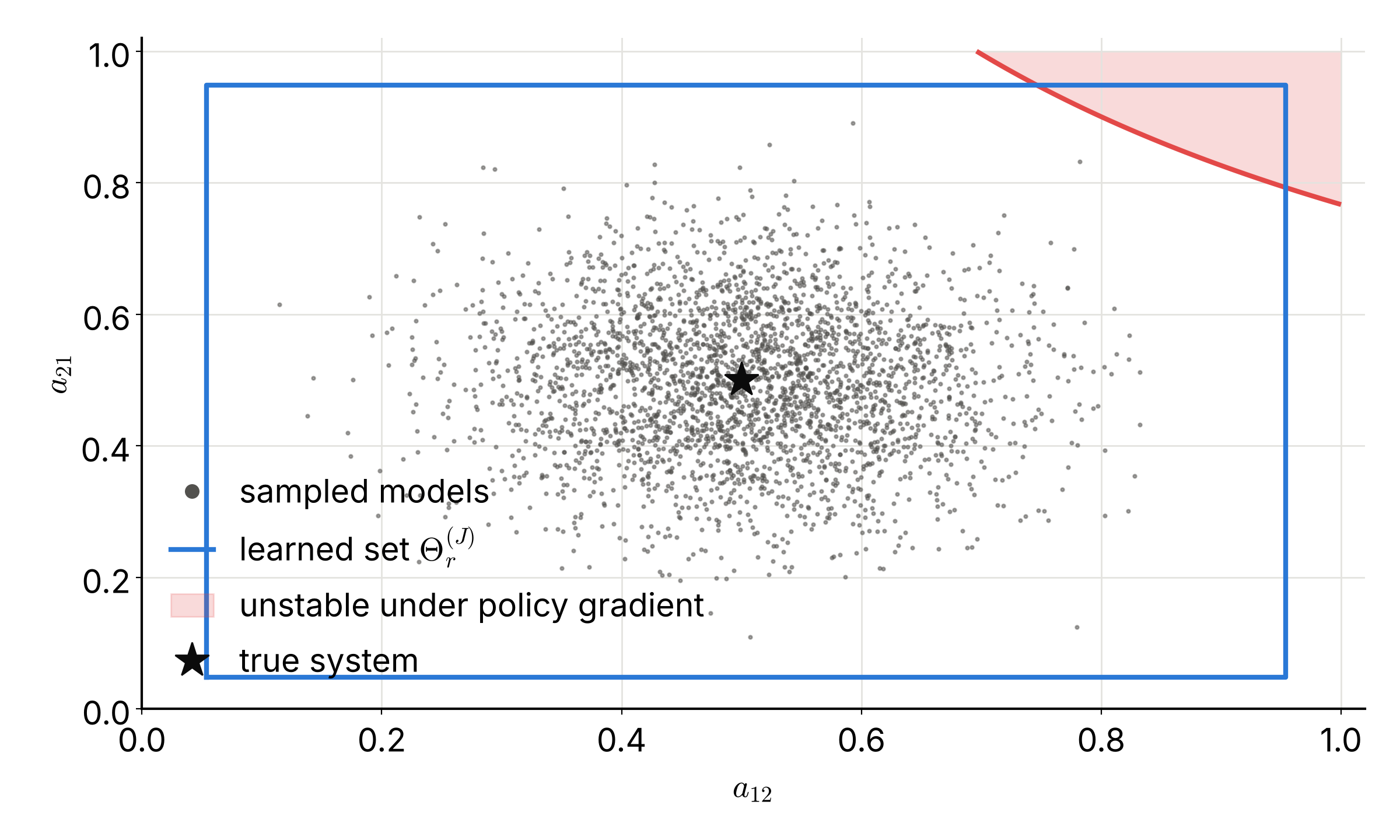}\\[-1mm]
\includegraphics[width=\columnwidth]{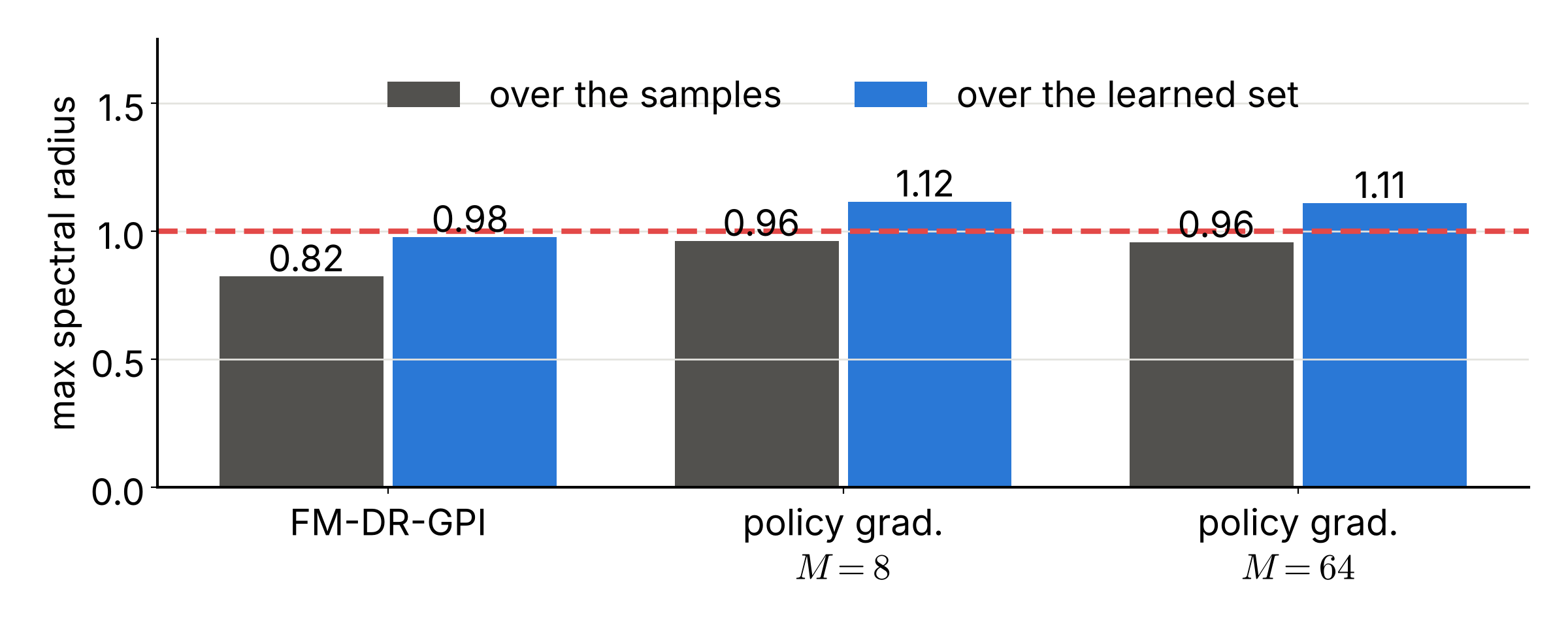}
\caption{Top: learned set in the $(a_{12},a_{21})$ plane, models sampled
from $P_{\rm des}$, and the region where the policy-gradient gain is
destabilizing (shaded; the other entries are at the center of the set).
Bottom: largest closed-loop spectral radius over the $20000$ sampled
models (grey) and over the learned set (blue), for the final gains of
FM-DR-GPI and of the policy gradient of \nw{\cite{nguyenle2026globally}}; the dashed line is
the stability limit.}
\label{fig:ex2}
\end{figure}


\section{Conclusion}
\label{sec:conclusion}
This paper developed Fisher-majorized domain-randomized generalized policy
iteration for noisy LQR with deterministic full-support stability.
Partial Bellman evaluation yields a structured Sylvester model, while Fisher curvature controls covariance mismatch from freezing the current covariance. The Fisher region certifies a safe step, and a one-dimensional SAA-cost search can extend it beyond the Fisher region. The analysis establishes SAA stationarity and transfers the guarantees to the population and true-plant costs, with
heterogeneity, certificate compression, sampling, and evaluation errors entering separately. Future work will consider adaptive certificate design and richer uncertainty decompositions to further reduce conservatism.

\appendices
\section{Proof of Lemma~\ref{thm:DR-robust-stability}}
\label{app:DR-robust-stability}
\begin{proof}
Recall from \eqref{eq:DR-closed-loop-GK} and \eqref{Gk-K} that
$A_K(\eta)=(\bar Y_1-N(\eta))G_K$ and $G_K=T_x+T_uK$.  With
$M_{\ell v}:=\bar Y_1-N_{\ell v}$, the closed-loop family over cell $\ell$ is
\(
\mathcal A_\ell(K)
:=\big\{MG_K:
M\in\operatorname{co}\{M_{\ell v}\}_{v=1}^{N_\ell^\theta}\big\}.
\)
Based on \eqref{DR-problem}, bundle $j$ certifies cell $\ell$ whenever
\begin{equation}
MG_K\mathcal X_\ell^{(j)}
\subseteq
\lambda_\ell^{(j)}\mathcal X_\ell^{(j)},
\qquad
\tfracorall M\in\operatorname{co}\{M_{\ell v}\}_{v=1}^{N_\ell^\theta}.
\label{eq:DR-cell-contraction}
\end{equation} 
We first show that if $\mathcal U_{\ell v}$ is not indexed by
$\widehat{\mathcal V}_\ell$ in \eqref{eq:DR-control-image-hull}, then  it is discarded. In fact, there exist
$\alpha_s\ge0$ with $\sum_s\alpha_s=1$ such that
\(
M_{\ell v}T_x
=
\sum_{s\in\widehat{\mathcal V}_\ell}
\alpha_s M_{\ell s}T_x,
\,
M_{\ell v}T_u
=
\sum_{s\in\widehat{\mathcal V}_\ell}
\alpha_s M_{\ell s}T_u.
\)
Consequently,
\(
M_{\ell v}G_K
=
\sum_{s\in\widehat{\mathcal V}_\ell}
\alpha_s M_{\ell s}G_K,
\, \tfracorall K.
\)
Therefore, 
\(
\operatorname{co}\{M_{\ell v}G_K:v\in[N_\ell^\theta]\}
=
\operatorname{co}\{M_{\ell v}G_K:v\in\widehat{\mathcal V}_\ell\}.
\)
Moreover, the right-hand side of \eqref{eq:DR-cell-contraction} is
$\lambda_\ell^{(j)}\mathcal X_\ell^{(j)}
=\{z:h_{\ell r}^{(j)\top}z\le\lambda_\ell^{(j)},
 r\in[m_{\ell,j}]\}$.  Hence, \eqref{eq:DR-cell-contraction} is equivalent
to $h_{\ell r}^{(j)\top}MG_Kx\le\lambda_\ell^{(j)}$ for every retained $M$, every $x\in\mathcal X_\ell^{(j)}$, and every facet $r$.  For fixed $r$, the left-hand side is affine in $M$ and linear in $x$, so its maximum is attained at a retained uncertainty image $M_{\ell v}$ and a certificate vertex $\xi_{\ell q}^{(j)}$.  The inclusion is therefore equivalent to $h_{\ell r}^{(j)\top}M_{\ell v}G_K\xi_{\ell q}^{(j)}
\le\lambda_\ell^{(j)}$ for all $i=(\ell,v,r,q)\in\mathcal I_j^{\rm raw}$.
Using $G_K=T_x+T_uK$ and
$a^\top Kb=(b\otimes a)^\top\vecop(K)$ converts these conditions
into
\eqref{eq:DR-raw-controller-polyhedron}.
Finally, $0\in\operatorname{int}\mathcal X_\ell^{(j)}$ and
$\lambda_\ell^{(j)}<1$ imply strict contractive invariance.  Standard
set-theoretic stability results~\cite{Blanchini2008} then imply Schur
stability. 
\end{proof}

\section{Proof of Proposition~\ref{prop:DR-normal-fan-reduction}}
\label{app:DR-normal-fan-reduction}
\begin{proof}
Based on the proof of Lemma~\ref{thm:DR-robust-stability}, for fixed
$(\ell,v,r)$ and $K$, the   certificate requires
\(
h_{\ell r}^{(j)\top}M_{\ell v}G_Kx\le\lambda_\ell^{(j)},
\,
\tfracorall x\in\mathcal X_\ell^{(j)}
\), or
\(
\max_{x\in\mathcal X_\ell^{(j)}}
h_{\ell r}^{(j)\top}M_{\ell v}G_Kx
\le\lambda_\ell^{(j)}.
\)
Since
\(
h_{\ell r}^{(j)\top}M_{\ell v}G_Kx
=
c_{\ell vr}^{(j)}(K)^\top x,
\)
the left-hand side is the maximum of a linear function over
$\mathcal X_\ell^{(j)}$ in the direction $c_{\ell vr}^{(j)}(K)$.

In general, a vertex $\xi_{\ell q}^{(j)}$ maximizes $c^\top x$ over
$\mathcal X_\ell^{(j)}$ if and only if
$c\in {\mathsf N}_{\mathcal X_\ell^{(j)}}(\xi_{\ell q}^{(j)})$.
Hence, for any $\vecop(K)\in\mathcal K_0$, the maximizing vertices are those
whose normal cones contain $c_{\ell vr}^{(j)}(K)\in
\mathcal D_{\ell vr}^{(j)}$. By the definition of
$\widehat{\mathcal Q}_{\ell vr}^{(j)}$ in \eqref{newQ}, these vertices are precisely the
retained ones, which proves
\eqref{eq:DR-normal-fan-support-equality}.
\end{proof}

\section{Proof of Lemma~\ref{prop:DR-parent-geometry}}
\label{app:DR-parent-geometry}
\begin{proof}
Since $\operatorname{int}\mathcal C_j\neq\varnothing$, there exists
$a\in\mathcal C_j$ and $r>0$ such that
$\mathbb B_2(a,r)\subseteq\mathcal C_j$. Hence, the LP in
\eqref{eq:DR-inner-anchor-LP} is feasible with a positive objective value,
so its optimum satisfies $r_j>0$. Moreover,
$\mathcal C_j\subseteq\mathcal K_0$ by \eqref{Cfinal1}. Since $\mathcal K_0$
is compact and $\mathcal C_j$ is closed, $\mathcal C_j$ is compact.
{Since $\mathcal C_j$ is compact, for every $d\ne0$, at least
one defining row must satisfy $G_{j,i}d>0$. Indeed, if
$G_{j,i}d\leq0$ for all $i\in\mathcal I_j^+$, then, for any
$k\in\mathcal C_j$ and $\alpha\geq0$,
\(
G_{j,i}(k+\alpha d)
=G_{j,i}k+\alpha G_{j,i}d
\leq h_{j,i}.
\)
Thus, the entire unbounded ray
$\{k+\alpha d:\alpha\geq0\}$ would belong to $\mathcal C_j$,
contradicting its compactness.}

We first show \eqref{eq:DR-parent-radial-reach-explicit}. Along the ray $k+\alpha d$, row $i$ requires
\(
G_{j,i}k+\alpha G_{j,i}d\le h_{j,i}.
\)
Since $k\in\mathcal C_j$, the initial slack
$h_{j,i}-G_{j,i}k$ is nonnegative. Rows with $G_{j,i}d\le0$ impose no upper
bound on $\alpha\ge0$, whereas each row with $G_{j,i}d>0$ requires
\(
\alpha\le
\frac{h_{j,i}-G_{j,i}k}{G_{j,i}d}.
\)
Thus, the first boundary crossing gives
\eqref{eq:DR-parent-radial-reach-explicit}.

To prove the second claim, for fixed $(a,r)$, the constraints in
\eqref{eq:DR-inner-anchor-LP} are equivalent to
\(
G_{j,i}(a+e)\le h_{j,i},
\,
\tfracorall\,\|e\|_2\le r,
\)
because
$\max_{\|e\|_2\le r}G_{j,i}e=r\|G_{j,i}\|_2$.
{Therefore, the LP computes the largest Euclidean ball contained in
$\mathcal C_j$.  Since $r_j>0$, its center satisfies
$a_j\in\operatorname{int}\mathcal C_j$.}
\end{proof}

\section{Proof of Proposition~\ref{thm:DR-bounded-complexity-inner-set}}
\label{app:DR-bounded-complexity-inner-set}
\begin{proof}
By {Lemma}~\ref{prop:DR-parent-geometry},
each radial first-hit point
\(
\kappa_{j,r}^{(0)}
=
a_j+\tau_{j,r}u_{j,r}
\)
belongs to $\mathcal C_j$.  Since $\mathcal C_j$ is convex,
$\widehat{\mathcal C}_{j,0}\subseteq\mathcal C_j$.
Since $a_j\in\operatorname{int}\mathcal C_j$, the maximum feasible step
$\tau_{j,r}=\bar\alpha_j(a_j,u_{j,r})$ is strictly positive.  Define
\(
\gamma_{j,r}
:=
({\tau_{j,r}^{-1}})/(
{\sum_{s=1}^{d_K+1}\tau_{j,s}^{-1}}).
\)
Then, $\gamma_{j,r}>0$, $\sum_r\gamma_{j,r}=1$, and, using
$\sum_r u_{j,r}=0$, one has
\(
\sum_{r=1}^{d_K+1}
\gamma_{j,r}\kappa_{j,r}^{(0)}
=
a_j+
\frac{\sum_{r=1}^{d_K+1}u_{j,r}}
{\sum_{s=1}^{d_K+1}\tau_{j,s}^{-1}}
=
a_j.
\)
The regular-simplex directions are affinely independent, and positive
scaling along these directions preserves affine independence of the resulting
points.  Thus, $a_j$ is a strict convex combination of
$d_K+1$ affinely independent vertices, so
$a_j\in\operatorname{int}\widehat{\mathcal C}_{j,0}$.
For any subsequent iteration, \eqref{eq:DR-inner-bundle-directional-margin}
maximizes the represented endpoint on the segment
$k_t+\alpha d$.  Since
$\widehat{\mathcal C}_{j,t}\subseteq\mathcal C_j$,
every represented endpoint is also feasible for the parent region, which gives
\(
\widehat\alpha_{j,t}(d)
\le
\alpha_{j,t}^{\rm ex}(d)
\).
If \eqref{eq:DR-inner-enrichment-trigger} holds, {Lemma}~\ref{prop:DR-parent-geometry}
gives
$\kappa_{j,t}^{\rm new}
=
k_t+\alpha_{j,t}^{\rm ex}(d)d\in\mathcal C_j$.
Hence, by convexity,
$\widehat{\mathcal C}_{j,t}^{+}\subseteq\mathcal C_j$.
The strict enrichment condition implies
$\kappa_{j,t}^{\rm new}\notin\widehat{\mathcal C}_{j,t}$, so
$\widehat{\mathcal C}_{j,t}\subsetneq
\widehat{\mathcal C}_{j,t}^{+}$.
Because $k_t\in\widehat{\mathcal C}_{j,t}$, the entire segment from
$k_t$ to $\kappa_{j,t}^{\rm new}$ lies in
$\widehat{\mathcal C}_{j,t}^{+}$, and therefore
\(
\widehat\alpha_{j,t}^{+}(d)
\ge
\alpha_{j,t}^{\rm ex}(d).
\)
The reverse inequality follows from
$\widehat{\mathcal C}_{j,t}^{+}\subseteq\mathcal C_j$, proving
\eqref{eq:DR-inner-enrichment-properties}.
Finally, the initialization stores exactly $d_K+1$ vertices, and every
successful enrichment adds one new vertex while no stored vertices are
removed.  Therefore,
\(
p_{j,t}
\le
d_K+1+
\bigl[p_{\max}-(d_K+1)\bigr]
=
p_{\max},
\)
so at most $p_{\max}-(d_K+1)$ enrichments can occur.
\end{proof}

\section{Proof of Lemma~\ref{lem:certified-pointwise-cost}}
\label{app:certified-pointwise-cost}
\begin{proof}
Since $u_k=Ky_k$ and \eqref{Sk}, the stage cost is 
\(
y_k^\top Qy_k+u_k^\top Ru_k
=
y_k^\top S_K y_k.
\)
Moreover, using $y_k=x_k+\upsilon_k$ and defining $X_K^{\rm true}(\eta):=\mathbb E[x_kx_k^\top]$, one has 
\(
\mathbb E[y_ky_k^\top]
=
X_K^{\rm true}(\eta)+V
\). Hence,
\(
J(K,\theta(\eta))
=
\operatorname{tr}\!\left(
S_K\,\mathbb E[y_ky_k^\top]
\right)
=
\operatorname{tr}\!\left(S_KX_K^{\rm true}(\eta)\right)
+
\operatorname{tr}\!\left(S_KV\right).
\) To form $X_K^{\rm true}(\eta)$, under the realization
$\theta(\eta)$ in \eqref{thetaN} and $u_k=Ky_k=K(x_k+\upsilon_k)$, the closed-loop dynamics are
\begin{align} \label{cllem1}
x_{k+1}
=
A_K(\eta)x_k+B(\eta)K\upsilon_k+\omega_k.
\end{align}
Thus, the overall disturbance is $B(\eta)K\upsilon_k+\omega_k$, with 
the covariance
\(
\Sigma_K(\eta)
:=
W+B(\eta)KVK^\top B(\eta)^\top.
\)
Then, $X_K^{\rm true}(\eta)=\mathbb E[x_kx_k^\top]$ is the
unique solution of
\begin{align} \label{covcl}
X_K^{\rm true}(\eta)
=
A_K(\eta)X_K^{\rm true}(\eta)A_K(\eta)^\top+\Sigma_K(\eta).
\end{align}

On the other hand, using \eqref{eq:DR-pointwise-P},
\(
S_K
=
P_K(\eta)-A_K(\eta)^\top P_K(\eta)A_K(\eta).
\) Thus, one has
\(
\operatorname{tr}\!\left(S_KX_K^{\rm true}(\eta)\right) \)
\(=
\operatorname{tr}\!\left(
P_K(\eta)\left[
X_K^{\rm true}(\eta)-A_K(\eta)X_K^{\rm true}(\eta)A_K(\eta)^\top
\right]
\right)=
\operatorname{tr}\!\left(P_K(\eta)\Sigma_K(\eta)\right),
\)
where the last equality is obtained from \eqref{covcl}.
Hence,
\(
J(K,\theta(\eta))
=
\operatorname{tr}\!\left(P_K(\eta)\Sigma_K(\eta)\right)
+\operatorname{tr}(S_KV).
\) By \eqref{eq:WV-certificates}, one has
\(
\Omega_K(\eta)-\Sigma_K(\eta)
=
\bar W-W
+
B(\eta)K(\bar V-V)K^\top B(\eta)^\top
\succeq0.
\)
Therefore,
\(
\bar J(K,\theta(\eta))-J(K,\theta(\eta))
=
\operatorname{tr}\!\left(
P_K(\eta)
[\Omega_K(\eta)-\Sigma_K(\eta)]
\right)\quad+
\operatorname{tr}\!\left(S_K(\bar V-V)\right)
\geq0,
\)
as $P_K(\eta)\succeq0$, $S_K\succeq0$, and
$\bar V-V\succeq0$, which proves \eqref{eq:pointwise-cost-bound}.

For the second part, the Part~I \cite{PartI} consistency results give
\(
\mathcal P_T\to\{\eta^\star\},
\,
\bar W_T\to W,
\,
\bar V_T\to V.
\) As the supports contract to the true model, the convergence is uniform over \(\mathcal P_T\), yielding \eqref{eq:pointwise-cost-asymptotic}.
\end{proof} \vspace{-12pt}

\section{Proof of Lemma \ref{lem:DR-dominating-Fisher}}
\label{proof:lem4}

\begin{proof}
Fix $\eta$ and $x_k$, and let $K(\varepsilon)=K_t+\varepsilon D$.
We seek to bound the Fisher information measuring the
sensitivity of $x_{k+1}\mid x_k$ to changes in $\varepsilon$
around zero. Under the Gaussian surrogate noises,
$\upsilon_k\sim\mathcal N(0,\bar V)$, the marginal transition is
$x_{k+1}\mid x_k\sim
\mathcal N(A_{K(\varepsilon)}(\eta)x_k,
\Omega_{K(\varepsilon)}(\eta))$,
which reduces to \eqref{eq:DR-Fisher-transition}
at $\varepsilon=0$.
To obtain a simpler upper bound, we condition additionally
on $\upsilon_k$, giving
$x_{k+1}\mid(x_k,\upsilon_k)\sim
\mathcal N(A_{K(\varepsilon)}(\eta)x_k+
B(\eta)K(\varepsilon)\upsilon_k,\bar W)$.
Unlike the marginal law, this conditional law has fixed
covariance $\bar W$ and mean derivative
$B(\eta)D(x_k+\upsilon_k)$.
Recall that the score is the derivative of the
log-likelihood with respect to $\varepsilon$, and
Fisher information is its expected square.
Thus, the conditional score at $\varepsilon=0$ is
\(
s_{\rm cond}=(x_{k+1}-\mu_0)^\top\bar W^{-1}
B(\eta)D(x_k+\upsilon_k),
\)
where $\mu_0=A_{K_t}(\eta)x_k+B(\eta)K_t\upsilon_k$.
Taking its expected square, using
$\mathbb E[(x_{k+1}-\mu_0)(x_{k+1}-\mu_0)^\top
\mid x_k,\upsilon_k]=\bar W$,
gives the conditional Fisher information.

Now consider the joint observation $(x_{k+1},\upsilon_k)$.
Its density factors as
$p_\varepsilon(x_{k+1},\upsilon_k\mid x_k)
=p_\varepsilon(x_{k+1}\mid x_k,\upsilon_k)p(\upsilon_k)$.
Since $p(\upsilon_k)$ is independent of $\varepsilon$,
the joint score equals $s_{\rm cond}$.
Hence, averaging the conditional Fisher information
over $\upsilon_k$ gives the joint Fisher information,
\(
\mathbb E_{\upsilon_k}[
(x_k+\upsilon_k)^\top D^\top B(\eta)^\top
\bar W^{-1}B(\eta)D(x_k+\upsilon_k)].
\)
The marginal score is the conditional expectation
of the joint score given $(x_{k+1},x_k)$.
By Jensen's inequality, its expected square cannot
exceed the joint Fisher information.
Finally, averaging over the stationary state $x_k$,
independent of $\upsilon_k$, and using
$\mathbb E[x_kx_k^\top]=X_t(\eta)$ and
$\mathbb E[\upsilon_k\upsilon_k^\top]=\bar V$, gives
$\mathbb E[(x_k+\upsilon_k)(x_k+\upsilon_k)^\top]
=X_t(\eta)+\bar V$. Therefore,
\(
\mathfrak I_{K_t}^{\rm tr}(\eta;D)
\le
\tr\!\left(
\bar W^{-1}B(\eta)D(X_t(\eta)+\bar V)
D^\top B(\eta)^\top
\right)=\mathfrak I_{K_t}(\eta;D),
\)
where the last equality uses
$\tr(D^\top MDC)=d^\top(C\otimes M)d$
for symmetric $C,M$ and $d=\operatorname{vec}(D)$.
\end{proof}

\section{Proof of Lemma~\ref{prop:DR-certificate-computable-bounds}}
\label{app:DR-certificate-computable-bounds}
\begin{proof}
Fix $j\in\mathcal J_t$ and $\eta\in\mathcal P_\ell$. Write
$A:=A_{K_t}(\eta)$,
$g(x):=g_{\ell j}(x):=\max_r h_{\ell r}^{(j)\top}x$, and
$\lambda:=\lambda_{\ell j}^{\rm cl}(K_t)$.
By \eqref{eq:DR-closed-loop-GK}, $A=M(\eta)G_{K_t}$,
where $M(\eta)$ lies in the convex hull of the vertex images
$M_{\ell v}$. Since $g$ is convex, the maximum of $g(Ax)$
over $\eta\in\mathcal P_\ell$ and
$x\in\mathcal X_\ell^{(j)}$ is attained at an uncertainty
image and a certificate vertex. Thus, by
\eqref{eq:DR-current-contraction},
$g(Ax)\le\lambda g(x)$, with
$\lambda\le\lambda_\ell^{(j)}$.
Moreover, \eqref{eq:DR-polyhedral-forcing} and
$\|K_t\|_2\le\bar k_2$ give
$\bar\omega_{\ell j}(K_t)\le\bar\omega_{\ell j}$.

Since
$r_{\ell j}\mathbb B_2\subseteq\mathcal X_\ell^{(j)}
\subseteq R_{\ell j}\mathbb B_2$, the gauge satisfies
$\|z\|_2/R_{\ell j}\le g(z)\le\|z\|_2/r_{\ell j}$.
Consequently,
\(
\|A^kx\|_2
\le R_{\ell j}\lambda^kg(x)
\le\frac{R_{\ell j}}{r_{\ell j}}
(\lambda_\ell^{(j)})^k\|x\|_2,
\)
which yields \eqref{eq:DR-certificate-power-bound}.

Next, the stationary covariance equation
\eqref{eq:DR-Fisher-stationary-X} gives, for any unit
vector $u$, 
\(
u^\top X_tu
=\sum_{k\ge0}
\max_{\|w\|_2\le1}
(u^\top A^k\Omega_{K_t}^{1/2}w)^2.
\)
With $y:=\Omega_{K_t}^{1/2}w$, the gauge contraction gives
$|u^\top A^ky|\le R_{\ell j}\lambda^kg(y)$, while
\(
g(y)\le
\max_r
\big(h_{\ell r}^{(j)\top}
\Omega_{K_t}h_{\ell r}^{(j)}\big)^{1/2}.
\)
By \eqref{eq:DR-Omega}, the map
\(
\eta\mapsto h^\top\Omega_{K_t}(\eta)h
=h^\top\bar Wh+
\|\bar V^{1/2}K_t^\top B(\eta)^\top h\|_2^2
\)
is convex, so its maximum over $\mathcal P_\ell$
is attained at a vertex $\eta_{\ell v}$.
Using \eqref{eq:DR-current-contraction} and summing
the geometric series gives the bound on $X_t$ in
\eqref{eq:DR-certificate-PX-bounds}. Its uniform
version follows from $\lambda\le\lambda_\ell^{(j)}$
and $\bar\omega_{\ell j}(K_t)\le\bar\omega_{\ell j}$.
Finally, the value equation \eqref{eq:DR-pointwise-P}
gives, for any unit vector $u$,
\(
u^\top P_{K_t}u
=\sum_{k\ge0}(A^ku)^\top S_{K_t}(A^ku).
\)
Since $z\mapsto z^\top S_{K_t}z$ is convex,
its maximum on $g(z)\mathcal X_\ell^{(j)}$
occurs at a scaled vertex. Hence, by
\eqref{eq:DR-polyhedral-forcing},
$z^\top S_{K_t}z\le g(z)^2\bar s_{\ell j}$.
Using
$g(A^ku)\le(\lambda_\ell^{(j)})^k/r_{\ell j}$
and summing the geometric series gives the stated
bound on $P_{K_t}$ in
\eqref{eq:DR-certificate-PX-bounds}.
\end{proof}

\section{Proof of Lemma \ref{lem:DR-covariance-distortion}}
\label{proof:lema7}
\begin{proof}
Fix $\eta$ and write
$B:=B(\eta)$, $F:=A_K(\eta)$,
$F':=A_{K+D}(\eta)=F+BD$,
$X:=X_K(\eta)$, $X':=X_{K+D}(\eta)$,
$\mathfrak I:=\mathfrak I_K(\eta;D)$,
$\xi:=\xi_K(\eta)$, and $\phi:=\phi_K(\eta;d)$.
Both $F$ and $F'$ are Schur because $K,K+D\in\mathcal S_{\rm cert}$.

Collect the state and measurement noise in
$Z:=\operatorname{blkdiag}(X,\bar V)$ and define
$\breve F:=[\,F\ \ BK\,]$ and $\mathsf E:=[\,I_n\ \ I_n\,]$.
Then, $\mathsf EZ\mathsf E^\top=X+\bar V$ and, by
\eqref{eq:DR-Fisher-stationary-X} and \eqref{eq:DR-Omega},
$X=\breve FZ\breve F^\top+\bar W$. Moreover,
$[\,F'\ \ B(K+D)\,]=\breve F+BD\mathsf E$. Hence, applying the
new-policy covariance operator to $X$ gives
\[
\begin{aligned}
F'XF'^\top+\Omega_{K+D}(\eta)
&=(\breve F+BD\mathsf E)Z(\breve F+BD\mathsf E)^\top+\bar W\\
&=X+\Delta,
\end{aligned}
\]
where
\[
\Delta:=BD\mathsf EZ\breve F^\top
+\breve FZ\mathsf E^\top D^\top B^\top
+BD(X+\bar V)D^\top B^\top.
\]
Subtracting this identity from
$X'=F'X'F'^\top+\Omega_{K+D}(\eta)$ yields
\begin{equation}
X'-X=\sum_{k\ge0}F'^k\Delta F'^{k\top}.
\label{eq:DR-covariance-difference}
\end{equation}

Define
\(
\mathsf U:=\bar W^{-1/2}BD(X+\bar V)^{1/2},
\mathsf V:=\bar W^{-1/2}\breve FZ\mathsf E^\top
(X+\bar V)^{-1/2}.
\)
Then,
\begin{align} \label{Wdelat}
\bar W^{-1/2}\Delta\bar W^{-1/2}
=\mathsf U\mathsf U^\top+
\mathsf U\mathsf V^\top+
\mathsf V\mathsf U^\top
\end{align}
and
\begin{align} \label{U2}
\|\mathsf U\|_F^2
=\operatorname{tr}\!\left(
\bar W^{-1}BD(X+\bar V)D^\top B^\top\right)
=\mathfrak I.
\end{align}
Furthermore,
$Z^{1/2}\mathsf E^\top(\mathsf EZ\mathsf E^\top)^{-1}
\mathsf EZ^{1/2}$ is an orthogonal projection, and therefore
\[
\mathsf V\mathsf V^\top
\preceq\bar W^{-1/2}\breve FZ\breve F^\top\bar W^{-1/2}
=\bar W^{-1/2}(X-\bar W)\bar W^{-1/2}.
\]
Thus, using \eqref{zeta1}, $\|\mathsf V\|_2^2\le\xi-1$. Therefore, using this and \eqref{Wdelat} and \eqref{U2}, one has
\[
\|\bar W^{-1/2}\Delta\bar W^{-1/2}\|_2
\le\mathfrak I+2\sqrt{\mathfrak I(\xi-1)}=\phi.
\]
Equivalently,
$-\phi\bar W\preceq\Delta\preceq\phi\bar W$.
Since $\bar W\preceq\Omega_{K+D}(\eta)$ and
$X'=\sum_{k\ge0}F'^k\Omega_{K+D}(\eta)F'^{k\top}$,
\eqref{eq:DR-covariance-difference} gives
\(
-\phi X'\preceq X'-X\preceq\phi X'.
\)
Because $\phi<1$, rearranging proves
\eqref{eq:DR-raw-covariance-sandwich}.
\end{proof}

\section{Proof of Theorem~\ref{prop:DR-e-model-validity}}
\label{app:e-model-validity}
\begin{proof}
Fix $\eta\in\mathcal P_\ell$ and suppress the arguments $K$ and $\eta$ when
unambiguous. Let $X:=X_K(\eta)$, $\Delta X:=X_{K+D}(\eta)-X$,
$a:=d^\top\mathbf A_K(\eta)d=\operatorname{tr}(D^\top\mathcal R_KDX)$,
$\mathfrak I:=\mathfrak I_K(\eta;D)$, $\psi:=\psi_K(\eta)$, and
$\phi:=\phi_K(\eta;d)$.
The identity of
Lemma~\ref{lem:DR-two-channel-PDL} holds for each $\eta$ before expectation:
$\bar J(K+D)-\bar J(K)=\operatorname{tr}(\Phi_KX_{K+D})+\operatorname{tr}(\Xi_K\bar V)$.
Expanding $\Phi_K$ and $\Xi_K$ and using
$\nabla_K\bar J(K,\theta(\eta))=2(\mathcal E_KX+\mathcal R_KK\bar V)$ gives
$\operatorname{tr}(\Phi_KX)+\operatorname{tr}(\Xi_K\bar V)
=\langle\nabla_K\bar J,D\rangle+\operatorname{tr}(D^\top\mathcal R_KD(X+\bar V))
=\mathfrak q_K(\eta;D)$. Hence,
\(
\bar J(K+D)-\bar J(K)-\mathfrak q_K(\eta;D)=\operatorname{tr}(\Phi_K\Delta X).
\)
By
Lemma~\ref{lem:DR-covariance-distortion},
$-\frac{\phi}{1+\phi}X\preceq\Delta X\preceq\frac{\phi}{1-\phi}X$. Hence
$\Delta X=X^{1/2}GX^{1/2}$ with $\|G\|_2\le c:=\phi/(1-\phi)$.
Now, write
$\Phi_K=D^\top\mathcal R_KD+(D^\top\mathcal E_K+\mathcal E_K^\top D)$. For the
first part, $|\operatorname{tr}(D^\top\mathcal R_KD\Delta X)|
=|\operatorname{tr}(X^{1/2}D^\top\mathcal R_KDX^{1/2}G)|\le c\,a$. For the
cross term,
$|\operatorname{tr}((D^\top\mathcal E_K+\mathcal E_K^\top D)\Delta X)|
=2|\operatorname{tr}(X^{1/2}\mathcal E_K^\top DX^{1/2}G)|
\le2c\,\|X^{1/2}\mathcal E_K^\top DX^{1/2}\|_*$, and
\(
\|X^{1/2}\mathcal E_K^\top DX^{1/2}\|_*
\le\|\mathcal R_K^{-1/2}\mathcal E_KX^{1/2}\|_F\,
\|\mathcal R_K^{1/2}DX^{1/2}\|_F
=\sqrt{\psi\,a},
\)
because $\mathcal E_K^\top\mathcal R_K^{-1}\mathcal E_K=\Psi_K$. Therefore
$|\operatorname{tr}(\Phi_K\Delta X)|\le c\,(a+2\sqrt{a\psi})$.
Since
$\phi\le\bar\phi_\ell$, $c\,a\le\frac{\bar\phi_\ell}{1-\bar\phi_\ell}a$, which
is the first correction term in \eqref{eq:DR-e-modified-curvature}. For the
cross term, $\phi=\sqrt{\mathfrak I}\,(\sqrt{\mathfrak I}+2\sqrt{\xi_K-1})$ and
$\sqrt{\mathfrak I}\le\sqrt{2\Delta_\ell^{\rm u}(K)}$ by
\eqref{eq:DR-full-Fisher-envelope-bound}, so
\(
2c\sqrt{a\psi}\le\frac{\sqrt\psi\,\big(\sqrt{2\Delta_\ell^{\rm u}(K)}
+2\sqrt{\xi_K-1}\big)}{1-\bar\phi_\ell}\;2\sqrt{\mathfrak I\,a}.
\)
The prefactor is at most $\beta_K(\eta)$ in
\eqref{eq:DR-residual-weighted-Fisher-beta}, since the regularizers only enlarge
it, and $2\sqrt{\mathfrak Ia}\le\varsigma\mathfrak I+\varsigma^{-1}a$. Hence,
$2c\sqrt{a\psi}\le\beta_K(\eta)(\varsigma\mathfrak I+\varsigma^{-1}a)$. The
residual weight $\sqrt\psi$ thus arises from the cross term between the Riccati
residual and the covariance error, and the Fisher form arises from the
covariance error itself.
Finally, adding the two bounds and using
$\mathfrak I=d^\top\mathbf F_Kd$, $a=d^\top\mathbf A_Kd$, and
\eqref{eq:DR-e-modified-curvature} gives
$|\operatorname{tr}(\Phi_K\Delta X)|\le\frac12d^\top(\mathbf H_K^{\rm e}(\eta)-\mathbf M_K(\eta))d$,
which is \eqref{eq:DR-Fisher-relative-model-gap}. Averaging over the batch and
using $\mathfrak q_K+\frac12d^\top(\mathbf H^{\rm e}-\mathbf M)d=\mathfrak m^{\rm e}$
gives the one-sided bound for $f_N$, and averaging over $\varrho$ gives it
for $f$.
\end{proof}

\section{Proof of Proposition \ref{prop:DR-certified-upper-model}}
\label{proof:prop3}
\begin{proof}
Fix $d$ with $K_t+\operatorname{mat}(d)\in\mathcal S_{\rm cert}$
and $\Gamma_{K_t}(d)\le1$, and set $D=\operatorname{mat}(d)$.
For each batch realization $\eta_i$, Theorem~\ref{prop:DR-e-model-validity}
and \eqref{eq:DR-Fisher-relative-model-gap} give
\(
\bar J(K_t+D,\theta(\eta_i))
-\bar J(K_t,\theta(\eta_i))
\le
\mathfrak q_{K_t}(\eta_i;D)
+\tfrac12d^\top
(\mathbf H_{K_t}^{\rm e}(\eta_i)
-\mathbf M_{K_t}(\eta_i))d
=
\langle\nabla_K\bar J(K_t,\theta(\eta_i)),D\rangle
+\tfrac12d^\top\mathbf H_{K_t}^{\rm e}(\eta_i)d,
\)
where the equality follows by substituting the definition of
$\mathfrak q_{K_t}$ in Theorem~\ref{prop:DR-e-model-validity}.
Averaging over $i$ and using the definitions of
$f_N$, $\mathbf g_N(K_t)$, and $\mathbf H_{N,K_t}^{\rm e}$ yields
\(
f_N(K_t+D)-f_N(K_t)
\le\mathbf g_N(K_t)^\top d
+\tfrac12d^\top\mathbf H_{N,K_t}^{\rm e}d
=\mathfrak m_{N,K_t}^{\rm e}(d).
\)

Next, by \eqref{eq:DR-Fisher-direction-model} and
\eqref{eq:DR-SAA-e-model},
\(
\mathcal U_t(d)-\mathfrak m_{N,K_t}^{\rm e}(d)
={}(\widetilde{\mathbf g}_t-\mathbf g_N(K_t))^\top d
+\epsilon_t^g\|d\|_2
+\tfrac12d^\top
(\mathbf H_t^{\rm e}-\mathbf H_{N,K_t}^{\rm e})d.
\)
Lemma~\ref{lem:DR-SAA-evaluation-error} and Cauchy--Schwarz imply
that the first two terms lie in
$[0,2\epsilon_t^g\|d\|_2]$.
Moreover, the same lemma and
$\mathbf H_t^{\rm e}
=\widetilde{\mathbf H}_t^{\rm e}+\epsilon_t^HI$
give
\(
0\preceq
\mathbf H_t^{\rm e}-\mathbf H_{N,K_t}^{\rm e}
\preceq2\epsilon_t^HI.
\)
Thus, the last term lies in
$[0,\epsilon_t^H\|d\|_2^2]$, proving
\eqref{eq:DR-computed-e-model-error} and hence the
second inequality in \eqref{eq:DR-SAA-majorization}.
\end{proof}

\section{Proof of Lemma~\ref{lem:DR-two-channel-PDL}}
\label{app:DR-two-channel-PDL}
\begin{proof}
Fix $\eta$ and write
$P:=P_K(\eta)$, $F:=A(\eta)+B(\eta)K$,
$F':=A(\eta)+B(\eta)K'$, $X:=X_K(\eta)$, and
$X':=X_{K'}(\eta)$.    
{From the surrogate-cost definition in
Lemma~\ref{lem:certified-pointwise-cost},}
\[
\bar J(K')-\bar J(K)
=
\tr(S_{K'}X'-S_KX)
+\tr((S_{K'}-S_K)\bar V).
\]
Using the current-policy value and covariance equations
$P=S_K+F^\top PF$ and $X=FXF^\top+\Omega_K$ gives
$\tr(S_KX)=\tr(P\Omega_K)$. Likewise, the new-policy covariance
equation $X'=F'X'F'^\top+\Omega_{K'}$ gives
$\tr(P\Omega_{K'})=\tr((P-F'^\top PF')X')$.
Therefore,
\begin{align*}
& \bar J(K',\theta(\eta))-\bar J(K,\theta(\eta))
=\tr\!\left[
\big(S_{K'}+F'^\top PF'-P\big)X'\right]\\
&\quad+\tr\!\left[P(\Omega_{K'}-\Omega_K)\right]
+\tr\!\left[(S_{K'}-S_K)\bar V\right].
\end{align*}
Expand $F'=F+B\Delta$ and
$S_{K'}-S_K=\Delta^\top RK+K^\top R\Delta+\Delta^\top R\Delta$.
Collecting the terms multiplying $X'$ yields
\(
\Phi_K(\Delta)
=\Delta^\top\mathcal R_K\Delta
+\Delta^\top\mathcal E_K
+\mathcal E_K^\top\Delta.
\)
The remaining terms combine, using
$\Omega_K=\bar W+B K\bar V K^\top B^\top$, into
\(
\Xi_K(\Delta)
=\Delta^\top\mathcal R_K\Delta
+\Delta^\top\mathcal R_KK
+K^\top\mathcal R_K\Delta.
\)
Completing the square with
$D_K^{\rm H}=-\mathcal R_K^{-1}\mathcal E_K$ gives
\(
\Phi_K(\Delta)
=(\Delta-D_K^{\rm H})^\top\mathcal R_K(\Delta-D_K^{\rm H})-\Psi_K,
\)
and completing the second square gives
\(
\Xi_K(\Delta)
=(\Delta+K)^\top\mathcal R_K(\Delta+K)-K^\top\mathcal R_KK.
\)
Substitution proves the realization-wise identity; taking
$\mathbb E_{\eta\sim\varrho}$ gives \eqref{eq:DR-two-channel-PDL}.
\end{proof}

\section{Proof of Lemma~\ref{lem:DR-dispersion-identities}}
\label{app:DR-dispersion-identities}
\begin{proof}
(a) Using $\mathcal R_KD_K^{\rm H}=-\mathcal E_K$ and
$D_K^{{\rm H}\top}\mathcal R_KD_K^{\rm H}=\Psi_K$, the two squared terms in
$\mathfrak h(K;D)$ expand, for each $\eta$, as
\[
\begin{aligned}
&\operatorname{tr}(D^\top\mathcal R_KDX_K)
+2\operatorname{tr}(D^\top\mathcal E_KX_K)
+\operatorname{tr}(\Psi_KX_K),\\
&\operatorname{tr}(D^\top\mathcal R_KD\bar V)
+2\operatorname{tr}(D^\top\mathcal R_KK\bar V)
+\operatorname{tr}(K^\top\mathcal R_KK\bar V).
\end{aligned}
\]
The constant terms sum to the integrand of $\mathfrak n(K)$. The linear terms
equal $\langle2(\mathcal E_KX_K+\mathcal R_KK\bar V),D\rangle$, whose
expectation is $\langle\nabla f(K),D\rangle$ because
$\mathcal E_KX_K+\mathcal R_KK\bar V
=\mathcal R_KK(X_K+\bar V)+B^\top P_KAX_K$. The quadratic terms equal
$\operatorname{tr}(D^\top\mathcal R_KD(X_K+\bar V))$, whose expectation is
$\frac12\langle D,\mathcal M_K[D]\rangle$. Taking expectations gives
$\mathfrak h(K;D)=\mathfrak n(K)+\mathfrak q_K(D)$, which is \eqref{gK}.
(b) By (a), $\mathfrak h(K;\cdot)=\mathfrak n(K)+\mathbf g(K)^\top d
+\frac12d^\top\mathbf M_Kd$. Since $\mathbf M_K\succ0$, its unique minimizer is
$d_K^\star=-\mathbf M_K^{-1}\mathbf g(K)$, with minimum value
$\mathfrak h^\star(K)=\mathfrak n(K)-\delta^\star(K)$. Completing the square
gives \eqref{eq:DR-dispersion-completion}, and rearranging the minimum value
gives \eqref{eq:DR-unregularized-decrement}.
(c) Both squared terms in $\mathfrak h$ are nonnegative, and
$\mathfrak h^\star(K)=\mathfrak h(K;D_K^\star)$. Hence $\mathfrak h^\star(K)=0$
if and only if both vanish at $D_K^\star$ for $\varrho$-almost every $\eta$.
Since $\mathcal R_K\succ0$ and $X_K\succeq\bar W\succ0$, the first vanishes if
and only if $D_K^\star=D_K^{\rm H}(\eta)$, and the second if and only if
$(D_K^\star+K)\bar V=0$.
(d) Since $\mathbf H_K^{\rm e}\succeq\mathbf M_K\succ0$, completing the square
in \eqref{eq:DR-population-e-model} gives \eqref{eq:DR-e-model-completion}.
Moreover, $\mathbf H_K^{\rm e}\succeq\mathbf M_K$ implies
$(\mathbf H_K^{\rm e})^{-1}\preceq\mathbf M_K^{-1}$, so
$0\le\delta_{\rm e}(K)\le\delta^\star(K)$. If
$\mathbf H_K^{\rm e}\preceq\kappa_{\rm e}\mathbf M_K$, then
$(\mathbf H_K^{\rm e})^{-1}\succeq\mathbf M_K^{-1}/\kappa_{\rm e}$, so
$\delta_{\rm e}(K)\ge\delta^\star(K)/\kappa_{\rm e}$.
\end{proof}

\section{Proof of Theorem~\ref{thm:DR-SAA-convergence}}
\label{app:DR-SAA-convergence}
\begin{proof}
The proof has 6 steps: 1) By
Proposition~\ref{prop:DR-certified-descent}(c),
$f_N(K_{t+1})\le f_N(K_t)-\Delta_t^{\rm c}$ with $\Delta_t^{\rm c}\ge0$. Summing
and using $f_N\ge0$ gives $\sum_{t\ge0}\Delta_t^{\rm c}\le f_N(K_0)<\infty$, so
$\Delta_t^{\rm c}\to0$.
2) Let
$\alpha_{\min}:=\min\{1,r_{\rm F}/\bar D,\sigma_0/\bar h\}>0$, and fix $t$ with
$d_t\neq0$. If $\epsilon_t^g\le\sigma_0\|d_t\|_2$, then
Proposition~\ref{prop:DR-certified-descent}(a) gives
$s_t\ge2\sigma_0\|d_t\|_2^2-\epsilon_t^g\|d_t\|_2\ge\sigma_0\|d_t\|_2^2>0$.
Since $h_t=d_t^\top\mathbf H_t^{\rm e}d_t\le\bar h\|d_t\|_2^2$, this yields
$s_t/h_t\ge\sigma_0/\bar h$. By \eqref{eq:DR-Fisher-radius},
$\alpha_t^{\rm F}\ge r_{\rm F}/\|d_t\|_2\ge r_{\rm F}/\bar D$. Hence,
\eqref{eq:DR-certified-step-length} gives $\alpha_t^{\rm c}\ge\alpha_{\min}$,
and Proposition~\ref{prop:DR-certified-descent}(a) gives
$\Delta_t^{\rm c}\ge\frac12\alpha_t^{\rm c}s_t
\ge\frac12\alpha_{\min}\sigma_0\|d_t\|_2^2$. If instead
$\epsilon_t^g>\sigma_0\|d_t\|_2$, then $\|d_t\|_2<\epsilon_t^g/\sigma_0$. In
both cases,
\(
\|d_t\|_2\le\max\Big\{\frac{\epsilon_t^g}{\sigma_0},
\Big(\frac{2\Delta_t^{\rm c}}{\alpha_{\min}\sigma_0}\Big)^{1/2}\Big\},
\)
and the right-hand side tends to zero by Step~1 and $\epsilon_t^g\to0$.
3) For $t\ge T_{\rm fix}$,
$\widehat{\mathcal C}_t=\widehat{\mathcal C}$. Problem
\eqref{eq:DR-Fisher-direction-predictor} is a strongly convex quadratic program
over a polytope, so its optimality condition is
\(
0=\widetilde{\mathbf g}_t+\mathbf H_t^{\rm e}d_t+n_t,\,
n_t\in\mathsf N_{\mathcal S}(k_t+d_t),
\)
with $n_t=0$ when $d_t=d_t^{\rm unc}$ is interior. By
\eqref{eq:DR-SAA-evaluation-error-bounds},
$\|\mathbf g_N(K_t)+n_t\|_2\le\|\mathbf g_N(K_t)-\widetilde{\mathbf g}_t\|_2
+\|\mathbf H_t^{\rm e}d_t\|_2\le\epsilon_t^g+\bar h\|d_t\|_2\to0$. Let
$K_{t_j}\to K_\star$. Then $k_{t_j}+d_{t_j}\to k_\star$ by Step~2, and
$n_{t_j}\to-\mathbf g_N(K_\star)$ by continuity of $\mathbf g_N$. Since
$\widehat{\mathcal C}$ is closed and convex, the graph of
$\mathsf N_{\mathcal S}$ is closed, so
$-\mathbf g_N(K_\star)\in\mathsf N_{\mathcal S}(k_\star)$, which is
\eqref{eq:DR-SAA-accumulation-stationarity}.
4) By Step~2, there is $t_1$
with $\|d_t\|_2\le r_{\rm F}$ for $t\ge t_1$. Then
$\alpha_t^{\rm F}\ge r_{\rm F}/\|d_t\|_2\ge1$ for all $t\ge t_1$ {, so the
Fisher check passes and, \nw{by the update rule in Section~\ref{subsubsec:DR-steps}}, the line search
is not used for $t\ge t_1$} 
5) Let the evaluation be
exact for $t\ge T_{\rm fix}$, so $\widetilde{\mathbf g}_t=\mathbf g_N(K_t)$,
$\mathbf H_t^{\rm e}=\mathbf H_{N,K_t}^{\rm e}$, and $\epsilon_t^g=0$. The
optimality condition in the direction $-d_t$, as in the proof of
Proposition~\ref{prop:DR-certified-descent}(a), gives
$s_t=-\mathbf g_N(K_t)^\top d_t\ge h_t\ge2\sigma_0\|d_t\|_2^2$, so
$s_t/h_t\ge1$ whenever $d_t\neq0$. With Step~4, $\alpha_t^{\rm c}=1$ for
$t\ge t_2:=\max\{T_{\rm fix},t_1\}$ {, and $\alpha_t=\alpha_t^{\rm c}=1$}, and
$K_{t+1}=K_t+\operatorname{mat}(d_t)$. (If $d_t=0$ for some such $t$, then
$K_t$ is stationary and the sequence is constant from then on.)
6) For $t\ge t_2$, write
$\phi:=f_N+\iota_{\mathcal S}$ and note that $k_{t+1}-k_t=d_t$. The iterates
satisfy the following three conditions:
a) (Sufficient decrease:) by
Proposition~\ref{prop:DR-certified-descent}(c) and $s_t\ge h_t$,
$\phi(k_t)-\phi(k_{t+1})\ge\Delta_t^{\rm c}=s_t-\frac12h_t\ge\frac12s_t
\ge\sigma_0\|k_{t+1}-k_t\|_2^2$.
(b) ({Relative error:}) the optimality condition gives
$n_t\in\mathsf N_{\mathcal S}(k_{t+1})$ with
$\mathbf g_N(K_t)+\mathbf H_{N,K_t}^{\rm e}d_t+n_t=0$. Hence
$\mathbf g_N(K_{t+1})+n_t\in\partial\phi(k_{t+1})$ and
$\|\mathbf g_N(K_{t+1})+n_t\|_2
\le\|\mathbf g_N(K_{t+1})-\mathbf g_N(K_t)\|_2+\|\mathbf H_{N,K_t}^{\rm e}d_t\|_2
\le(L_N+\bar h)\|k_{t+1}-k_t\|_2$, where $L_N$ is a Lipschitz constant of
$\mathbf g_N$ on the compact set $\mathcal S$.
(c) ({Continuity:}) $f_N$ is continuous on the closed set $\mathcal S$,
so $\phi(k_{t_j})\to\phi(k_\star)$ along every convergent subsequence.
The iterates are bounded because $\mathcal S$ is compact. Moreover, $\phi$ is a
KL function. For each realization, $\vecop(P_K(\eta_i))
=(I-A_K(\eta_i)^\top\otimes A_K(\eta_i)^\top)^{-1}\vecop(S_K)$, and similarly
for $X_K(\eta_i)$, so their entries are rational functions of $k$ whose
denominators do not vanish on the stabilizing set. Hence, $f_N$ is semialgebraic
on $\mathcal S$. Since $\mathcal S$ is a polytope, $\phi$ is semialgebraic and
therefore satisfies the KL property at every point. By~\cite[Th.~2.9]{Attouch2013Convergence},
the sequence $\{k_t\}$ has finite length and converges to a critical point of
$\phi$, that is, to a point satisfying
\eqref{eq:DR-SAA-accumulation-stationarity}.
\end{proof}

\bibliographystyle{IEEEtran}
\bibliography{refs1}

\begin{IEEEbiography}[{\includegraphics[width=.75in,height=.95in,clip,keepaspectratio]{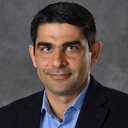}}]%
{\mdseries\normalfont Hamidreza Modares}
received the B.S., M.S., and Ph.D. degrees in electrical engineering from the University of Tehran, Shahrood University of Technology, and the University of Texas at Arlington, respectively. He is an Associate Professor at Michigan State University. His interests include safe learning-based control, multi-agent systems, and robotics.
\end{IEEEbiography}\vspace{-14pt}


\end{document}